\documentclass[journal]{IEEEtran}

\newif\ifarxiv
\arxivtrue

\usepackage{amsfonts,amsmath,amsthm,bm}
\usepackage{cite}
\usepackage{amsmath,amssymb,amsfonts}
\usepackage{algorithmic}
\usepackage{graphicx}
\usepackage{textcomp}
\usepackage{xcolor}
\usepackage[T1]{fontenc}
\usepackage{comment}
\usepackage{eucal}
\usepackage{algorithm}
\usepackage{algorithmic}
\usepackage[acronyms,nonumberlist,nopostdot,nomain,nogroupskip]{glossaries}

\usepackage{subcaption}
\usepackage[font=footnotesize,compatibility=false]{caption}
\usepackage{booktabs}
\usepackage{hyperref}
\hypersetup{
    colorlinks=true,        
    linkcolor=blue,         
    citecolor=blue,          
    urlcolor=black
}
\usepackage{orcidlink}
\usepackage{enumitem,flushend}

\usepackage{cuted}

\newtheorem{lem}{Lemma}
\newtheorem{theorem}{Theorem}

\newtheorem{remk}{Remark}

\newtheorem{prop}{Proposition}
\newtheorem{corol}{Corollary}

\newcommand{\ba}{\mathbf{a}}

\newcommand{\e}{\mathsf{e}}

\newcommand{\jj}{\mathsf{j}}
\newcommand{\bh}{\mathbf{h}}
\newcommand{\bH}{\mathbf{H}}
\newcommand{\bR}{\mathbf{R}}
\newcommand{\bx}{\mathbf{x}}
\newcommand{\bX}{\mathbf{X}}
\newcommand{\bn}{\mathbf{n}}

\newcommand{\by}{\mathbf{y}}

\newcommand{\bI}{\mathbf{I}}
\newcommand{\bu}{\mathbf{u}}
\newcommand{\bU}{\mathbf{U}}

\newcommand{\bq}{\mathbf{q}}

\newcommand{\Var}{\mathbb{V}}
\newcommand{\Cov}{\mathsf{Cov}}

\newcommand{\Herm}{\mathsf{H}}
\newcommand{\Trans}{\mathsf{T}}
\newcommand{\x}{\mathsf{x}}
\newcommand{\y}{\mathsf{y}}
\newcommand{\z}{\mathsf{z}}
\newcommand{\dd}{\mathsf{d}}
\newcommand{\bigo}{\mathcal{O}}
\newcommand{\defeq}{\triangleq}
\newcommand{\Ex}{\mathbb{E}}
\newcommand{\kk}{\kappa}
\newcommand{\erfi}{\operatorname{erfi}}

\def\bzero{\boldsymbol{0}}

\def\LOS{\mathrm{LOS}}
\def\Re{\mathrm{Re}}
\def\Im{\mathrm{Im}}

\def\SINR{\mathrm{SINR}}

\def\nLOS{\mathrm{NLOS}}
\def\tx{\mathrm{tx}}
\def\rx{\mathrm{rx}}
\def\vtx{\mathrm{vtx}}
\def\vrx{\mathrm{vrx}}
\def\NS{\mathrm{NS}}

\def\sinc{\mathrm{sinc}}
\def\max{\mathrm{max}}

\def\NF{\mathrm{NF}}
\def\refl{\mathrm{rfl}}
\def\scat{\mathrm{scr}}

\def\Cset{\mathbb{C}}
\def\Rset{\mathbb{R}}

\def\Uset{\mathbb{U}}

\def\sCN{\mathcal{CN}}
\def\Nset{\mathcal{N}}

\def\Uset{\mathcal{U}}

\usepackage{xcolor}
\newcommand{\MD}[1]{{\color{black} #1}}

\ifarxiv
\else
\usepackage{microtype}
\fi

\newacronym{NS}{NR}{non-ideal reflector}
\newacronym{RIS}{RIS}{reconfigurable intelligent surface}
\newacronym{QoS}{QoS}{quality of service}
\newacronym{LC}{LC}{liquid crystal}
\newacronym{SNR}{SNR}{signal-to-noise ratio}
\newacronym{SINR}{SINR}{signal-to-interference-and-noise ratio}
\newacronym{SIR}{SIR}{signal to interference ratio}
\newacronym{SMR}{SMR}{side lobe to main lobe ratio}
\newacronym{TDMA}{TDMA}{time-division multiple-access}
\newacronym{FDMA}{FDMA}{frequency-division multiple-access}
\newacronym{SDMA}{SDMA}{space-division multiple-access}
\newacronym{BS}{BS}{base station}
\newacronym{MU}{UE}{user equipment}
\newacronym{NF}{NF}{near-field}
\newacronym{FF}{FF}{far-field}
\newacronym{Tx}{Tx}{transmitter}
\newacronym{Rx}{Rx}{receiver}
\newacronym{AWGN}{AWGN}{additive white Gaussian noise}
\newacronym{wrt}{w.r.t.}{with respect to}
\newacronym{RDE}{RDE}{reaction-diffusion equation}
\newacronym{PDE}{PDE}{partial differential equation}
\newacronym{UPA}{UPA}{uniform planar array}
\newacronym{ULA}{ULA}{uniform linear array}
\newacronym{AO}{AO}{alternative optimization}
\newacronym{SOCP}{SOCP}{second-order cone programming}
\newacronym{AoD}{AoD}{angle of departure}
\newacronym{AoA}{AoA}{angle of arrival}
\newacronym{PIN}{PIN}{positive-intrinsic-negative}
\newacronym{RF}{RF}{radio frequency}
\newacronym{MEMS}{MEMS}{micro-electro-mechanical system}
\newacronym{LOS}{LOS}{line-of-sight}
\newacronym{nLOS}{NLOS}{non-LOS}
\newacronym{MISO}{MISO}{multi-input single-output}
\newacronym{MIMO}{MIMO}{multiple-input multiple-output}
\newacronym{CSI}{CSI}{channel state information}
\newacronym{PDF}{PDF}{probability density function}
\newacronym{PCD}{PCD}{parallel coordinate descent}
\newacronym{SS}{SS}{surface scattering}
\newacronym{RV}{RV}{random variable}
\newacronym{HF}{HF}{Huygens-Fresnel}
\newacronym{SR}{SR}{surface reflector}
\newacronym{CLT}{CLT}{central limit theorem}
\newacronym{iid}{IID}{independent and identically distributed}
\newacronym{mmWave}{mmWave}{millimeter wave}
\newacronym{XL}{XL}{extremely large}
\newacronym{GD}{GD}{Gaussian distribution}
\newacronym{DoF}{DOF}{degree of freedom}
\newacronym{EDoF}{EDOF}{effective degree of freedom}
\newacronym{GO}{GO}{geometric optics}
\newacronym{SE}{SE}{spectral efficiency}
\newacronym{FR3}{FR3}{frequency range 3}
\newacronym{RT}{RT}{Ray Tracing}
\title{Near-Field Multipath MIMO Channels:\\ Modeling Reflectors and Exploiting NLOS Paths}

\author{
\IEEEauthorblockN{Mohamadreza Delbari$^{\orcidlink{0000-0002-4768-5874}}$, \textit{Student Member, IEEE}, 
George C. Alexandropoulos$^{\orcidlink{0000-0002-6587-1371}}$, \textit{Senior Member, IEEE},\\
Robert Schober$^{\orcidlink{0000-0002-6420-4884}}$, \textit{Fellow, IEEE}, 
H. Vincent Poor$^{\orcidlink{0000-0002-2062-131X}}$, \textit{Life Fellow, IEEE},  
and Vahid Jamali$^{\orcidlink{0000-0003-3920-7415}}$, \textit{Senior Member, IEEE}\\
}
\IEEEauthorblockA{
\thanks{The work of M. Delbari and V. Jamali was supported in part by the Deutsche Forschungsgemeinschaft (DFG, German Research Foundation) within the Collaborative Research Center MAKI (SFB 1053, Project-ID 210487104), in part by the LOEWE initiative (Hesse, Germany) within the emergenCITY Centre under Grant LOEWE/1/12/519/03/05.001(0016)/72, and in part by the German Federal Ministry for Research, Technology and Space (BMFTR) under the program of ``Souverän. Digital. Vernetzt.'' joint project Open6GHub plus (Project-ID 16KIS2407). The work of G. C. Alexandropoulos was supported by the SNS JU project TERRAMETA under Grant 101097101. The work of R. Schober was funded by the German Federal Ministry for Research, Technology and Space (BMFTR) under the program of ``Souverän. Digital. Vernetzt.'' joint project 6G-RIC (Project-ID 16KISK023) and the Deutsche Forschungsgemeinschaft (DFG, German Research Foundation) under projects SFB 1483 (Project-ID 442419336, EmpkinS). The work of H. V. Poor was supported in part by the U.S National Science Foundation under Grants CNS-2128448 and ECCS-2335876. Part of this paper was presented at the IEEE GLOBECOM 2024 [DOI: \href{https://ieeexplore.ieee.org/document/11101277}{10.1109/GCWkshp64532.2024.11101277}] in~\cite{delbari2024nearfield}.
(Corresponding author: Mohamadreza Delbari).}
\thanks{M. Delbari and V. Jamali are with the Resilient Communication Systems Laboratory, Technische Universität Darmstadt, 64283 Darmstadt, Germany (e-mail: mohamadreza.delbari@tu-darmstadt.de; vahid.jamali@tu-darmstadt.de).} 
\thanks{ G. C. Alexandropoulos is with the Department of Informatics and Telecommunications, National and Kapodistrian University of Athens, 16122 Athens, Greece and the Department of Electrical and Computer Engineering, University of Illinois Chicago, Chicago, IL 60601, USA (e-mail: alexandg@di.uoa.gr).} 
\thanks{R. Schober is with the Institute for Digital Communication, Friedrich-Alexander-Universität Erlangen-Nürnberg, 91052 Erlangen, Germany (e-mail: robert.schober@fau.de).}
\thanks{H. V. Poor is with the Department of Electrical and Computer Engineering, Princeton University, Princeton, NJ 08544 USA (email: poor@princeton.edu).}
}
}

\begin{document}

\maketitle
\begin{abstract}
\Gls{NF} communications is receiving renewed interest in the context of \gls{MIMO} systems involving large physical apertures with respect to the signal wavelength. While \gls{LOS} links are typically expected to dominate in \gls{NF} scenarios, the impact of \gls{nLOS} components in the centimeter- and millimeter-wave frequency bands may be, in general, non-negligible. Moreover, although weaker than the \gls{LOS} path, \gls{nLOS} links may be essential for achieving multiplexing gains in \gls{MIMO} systems. The commonly used \gls{NF} channel models for \gls{nLOS} links in the literature are based on the point scattering assumption, which is not valid for large reflectors such as walls, ceilings, and the ground. In this paper, we develop a generalized statistical \gls{NF} \gls{MIMO} channel model that extends the widely adopted point scattering framework to account for imperfect reflections from large surfaces. This model is then leveraged to investigate how the physical characteristics of these reflectors influence the resulting \gls{NF} \gls{MIMO} channel. In addition, using the proposed channel model, we analytically demonstrate for a multi-user scenario that, even when users are located within the \gls{NF} regime, relying solely on \gls{LOS} \gls{NF} links may be insufficient to achieve multiplexing gains, thus exploiting \gls{nLOS} links becomes essential. Our simulation results validate the accuracy of the proposed model and show that, in many practical settings, the contribution of \gls{nLOS} components is non-negligible and must be carefully accounted for in the system design.

\end{abstract}
\glsresetall

\section{Introduction}
\label{Introduction}
The growing demand for higher data rates is pushing wireless systems toward higher frequency bands, such as centimeter and \gls{mmWave}, to exploit the abundant available bandwidth. However, a fundamental challenge at these frequencies is the significant propagation loss. To address this and maintain a reliable link budget, large antenna arrays can be exploited to enable the high beamforming gains required for reliable communications \cite{bjornson2020power,Xu2024,Shakya2024}. The combination of shorter wavelengths and larger antenna apertures greatly increases the Fraunhofer distance, thereby pushing a significant portion of the communication range into the \gls{NF} region \cite{Gong2024}. In this region, the conventional far-field plane-wave approximation becomes inaccurate, necessitating the adoption of more precise spherical-wave channel models.

Several studies have explored \gls{MIMO} wireless systems operating in the \gls{NF} regime; see \cite{liu2023near} for a comprehensive tutorial and \cite{Wei2023} for a tri-polarized channel model. For instance, under a \gls{LOS} \gls{NF} channel model, beamforming and localization have been investigated in \cite{delbari2026fast,Wei2022} and \cite{Ebadi2025,Gavras2025}, respectively. Furthermore, the authors of \cite{ramezani2023near} demonstrated that optimal multiplexing is possible in the \gls{NF} regime even under \gls{LOS} conditions. However, \gls{LOS} links are not always available, e.g., due to self-blockage. Consequently, \gls{nLOS} links must also be exploited to improve multiplexing gains in multi-user \gls{MIMO} systems. An \gls{NF} \gls{nLOS} channel model based on point scattering was proposed in \cite{liu2023near,lu2023near}. Nonetheless, experimental measurements indicate that dominant \gls{nLOS} components are often caused by large surfaces in the environment, such as walls, ceilings, and ground, which all possess substantial electrical apertures \cite{jansen2011diffuse,sheikh2021scattering}. We refer to such surfaces as \glspl{NS}; these \glspl{NS} cannot be accurately modeled as point scatterers. In particular, the interaction of electromagnetic waves with \glspl{NS} may give rise to both specular and non-specular (i.e., scattered) components. \MD{The existing rough-surface scattering models in the wireless communications literature predominantly focus on the far-field regime \cite{khamse2023scattering,grossman2017submillimeter,jansen2011diffuse,sheikh2021scattering}. On the other hand, studies addressing \gls{NF} propagation \cite{shi2017recovery,Demirci2026,Deutschmann2022} do not characterize the stochastic spatial correlation \cite{Duong2011} and \gls{MIMO} channel behavior in multi-antenna configurations. Capturing these statistical characteristics is essential, as the resulting spatial correlation and channel structure provide fundamental design insights for \gls{NF} beam focusing, beam tracking, and multi-user multiplexing.}

To the best of the authors' knowledge, a statistical \gls{NF} \gls{MIMO} channel model accounting for \glspl{NS} has not yet been reported in the literature. Addressing this gap is the main focus of this paper. The key contributions are summarized as follows.
\begin{itemize}
    \item We develop a novel \gls{NF} \gls{MIMO} channel model that captures the impact of \gls{nLOS} propagation paths resulting from imperfect reflections at large and rough \glspl{NS}. The model parameters are characterized as functions of the surface roughness variance, carrier frequency, \gls{AoD}, and \gls{AoA}. Unlike the conference version of this paper's framework \cite{delbari2024nearfield}, we also incorporate the effect of surface-length correlation to provide a more realistic characterization of the channel.
    \item The \gls{nLOS} channel matrix is decomposed into the sum of a deterministic component (i.e., the channel mean) and a stochastic component. Using principles from geometric optics, we demonstrate that, regardless of the roughness variance of the \gls{NS}, the deterministic component always corresponds to the channel matrix of a virtual \gls{LOS} link created by an ideal \gls{SR}. Furthermore, by applying the multidimensional \gls{CLT}, we show that the stochastic component follows a multivariate Gaussian distribution. Moreover, we derive theoretical expressions for the covariance matrix of this distribution for several special cases. This stochastic term accounts for the effect of \gls{SS}. Unlike \cite{delbari2024nearfield}, we also derive a theoretical expression quantifying the impact of the surface-length correlation on the channel power gains.
    \item The proposed model provides several insights for system design. In particular, we reveal how both the \gls{LOS} and deterministic \gls{nLOS} components of the \gls{NF} channel depend on the positions of the \glspl{Rx} and \glspl{Tx} or their virtual positions mirrored on the \gls{NS} in the three-dimensional space. These insights highlight the importance of \gls{NF} beam focusing to support both \gls{LOS} and \gls{nLOS} links for reliable \gls{NF} communications. Furthermore, we present a generalized \gls{NF} \gls{MIMO} Rician channel model that, while simplified, captures all essential features of the channel.
    \item Using the proposed \gls{NF} model, we analytically quantify how exploiting \glspl{NS} can enhance the spatial diversity of the channel. This facilitates the application of \gls{SDMA} to serve concurrently multiple \glspl{MU}. In particular, our results quantify the trade-off between exploiting \gls{nLOS} and \gls{LOS} links in a challenging case study involving a \gls{BS} serving two \glspl{MU}, which have limited resolvability if only the \gls{LOS} links are considered.
    \item We verify the accuracy of the proposed channel models in a scenario involving a \gls{Tx}, an \gls{Rx}, and an \gls{NS} using numerical evaluations and validate their statistical characteristics based on an extensive set of simulations. In addition, we demonstrate that \gls{nLOS} links are essential for achieving multiplexing gain with \gls{SDMA} in multi-user scenarios.
\end{itemize}

The remainder of this paper is organized as follows. In Section~\ref{sec: System and Channel Models}, we describe the overall system setup and briefly introduce both existing and the proposed channel models. Section~\ref{sec: Channel Model for Non-ideal Surface Reflection} then provides a detailed formulation of our proposed channel model, followed by theoretical analysis for a special case in Section~\ref{sec: When Exploiting nLOS Paths Is Beneficial?}. Simulation results are provided in Section~\ref{sec: Performance Evaluation}. Finally, Section~\ref{sec: Conclusion} concludes the paper.

\textit{Notation:} Bold capital and small letters are used to denote matrices and vectors, respectively.  $(\cdot)^\Trans$ and $(\cdot)^\Herm$ denote the transpose and Hermitian, respectively. Moreover,
$[\bX]_{m,n}$ and $[\bx]_{n}$ denote the element in the $m$th row and $n$th column of matrix $\bX$ and the $n$th entry of vector $\bx$, respectively. $\mathcal{CN}(\boldsymbol{\mu},\boldsymbol{\Sigma})$ and  $\Nset(\boldsymbol{\mu},\boldsymbol{\Sigma})$ denote complex and real Gaussian random vectors with mean vector $\boldsymbol{\mu}$ and covariance matrix $\boldsymbol{\Sigma}$, respectively. $\Ex\{\cdot\}$, $\Var\{\cdot\}$, and $\Cov\{\cdot,\cdot\}$ represent expectation, variance, and covariance, respectively. Finally, $o(\cdot)$, $\bigo(\cdot)$, $\mathbb{R}$, and $\mathbb{C}$ represent the little-o notation, the big-O notation, the set of real numbers, and the set of complex numbers, respectively. 

\section{System and Channel Models}
\label{sec: System and Channel Models}
\begin{figure}[t]
    \centering
    \includegraphics[width=0.3\textwidth]{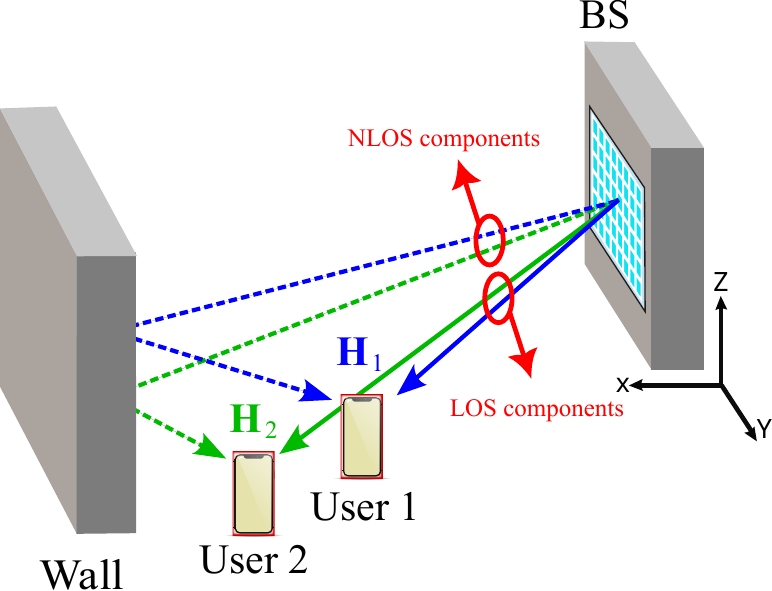}
    \caption{Schematic illustration of a downlink $K$-user wireless communication system for $K=2$ \glspl{MU}, whose channels comprise both LOS and NLOS links.}
    \label{fig:system model}
    \vspace{-0.5cm}
\end{figure}
In this section, we first present the considered system model. Subsequently, we introduce the existing \gls{MIMO} \gls{NF} channel model before presenting our proposed mixed-scattering \gls{NF} \gls{nLOS} channel model.

\subsection{System Model}
\label{sec: System model}
We consider a narrowband downlink communication system comprising a \gls{BS} with $N_\tx$ antenna elements and $K$ \glspl{MU}, each equipped with $N_\rx$ antenna elements, as shown in Fig.~\ref{fig:system model}. The received signal vector for the $k$th \gls{MU} ($k = 1, \dots, K$), $\by_k\in\Cset^{N_\rx}$, is expressed as
\begin{equation}\label{Eq:IRSbasic}
	\by_k = \bH_k\bx +\bn_k,
\end{equation}
where $\bx\in\Cset^{N_\tx}$ represents the transmit signal vector, subject to transmit power constraint $\Ex\{\|\bx\|^2\}\leq P_t$, with $P_t$ being the maximum transmit power. Assuming linear beamforming, the transmit vector $\bx\in\Cset^{N_\tx}$ can be written as $\bx=\sum_{k=1}^K\sqrt{P_k}\bq_ks_k$, where $\bq_k\in\Cset^{N_\tx}$, $s_k\in\Cset$, and $P_k\in\Rset^+$ are the normalized beamforming vector, the data symbol, and the transmit power allocated for the $k$th \gls{MU}, respectively. Here, $\Ex\{|s_k|^2\}=1$ and $\Ex\{s_ks_{k'}^*\}=0,\,\forall k'\neq k$ hold. The \gls{AWGN} at the $k$th \gls{MU} is denoted by $\bn_k\in\Cset^{N_\rx}$, where $\bn_k\sim\sCN(\bzero_{N_\rx},\sigma_n^2\bI_{N_\rx}),\,\forall k$, with $\sigma_n^2$ representing the noise power. Moreover, the channel matrix between the \gls{BS} and the $k$th \gls{MU} is denoted by $\bH_k\in\Cset^{N_\rx\times N_\tx}$. For notational simplicity, we omit the subscript $k$ in the subsequent analysis and define the channel model using a general matrix $\bH_{\NF}\in\Cset^{N_\rx\times N_\tx}$, where $N_\tx$ and $N_\rx$ represent the number of \gls{Tx} and \gls{Rx} antenna elements, respectively.

\subsection{Existing MIMO NF Channel Model}
\label{sec: Existing MIMO NF Channel Model}
In the \gls{NF} regime, the wavefront curvature across the \gls{Rx} plane becomes significant and cannot be ignored. Additionally, environmental scatterers cause multipath propagation, allowing the \gls{Rx} to receive signals reflected from scatterers through \gls{nLOS} paths. The following channel model has been widely adopted in the literature for \gls{MIMO} multipath channels in the \gls{NF} regime (see, for example, \cite[Eq. (25)]{liu2023near}, and \cite[Eq. (8)]{lu2023near}:
\begin{equation}\label{Eq:existingMIMO}
	\bH_{\NF} = c_0\bH_{\NF}^\LOS +
 \sum_{s=1}^S \hat{c}_s\bH^\scat_s,\quad
\end{equation}
where the channel matrices on the right-hand side are defined as follows ($m=1,\ldots,N_\rx$ and $n=1,\ldots,N_\tx$):
\begin{align} 
	[\bH_{\NF}^\LOS]_{m,n} &= \, \e^{\jj\kk\|\bu_{\rx,m}-\bu_{\tx,n}\|} \label{Eq:LoSnear},\\
    \bH^\scat_s&=\ba_{\rx}(\bu_{s})\ba_{\tx}^\Trans(\bu_{s}),\\
 {[\ba_{\tx}(\bu_{s})]_n}\! = \!\e^{\jj\kk\|\bu_{\tx,n}-\bu_{s}\|},
 \,\,\, 
 &{[\ba_{\rx}(\bu_{s})]_m}   \!=\! \e^{\jj\kk\|\bu_{\rx,m}-\bu_{s}\|}\!\!.\quad\label{Eq:nLoSnearPoint}
\end{align}
In \eqref{Eq:existingMIMO}, we decompose $\bH_\NF$ into its \gls{LOS} and \gls{nLOS} components without loss of generality. Each component is further separated into a scalar term that represents the channel power gain and a normalized matrix that defines the channel structure. For instance, $\bH_{\NF}^\LOS$ represents the normalized LOS NF channel matrix, while $c_0$ is the corresponding channel power gain. Here, $\bu_{\tx,n}=(x_{\tx,n},y_{\tx,n},z_{\tx,n})$ and $\bu_{\rx,m}=(x_{\rx,m},y_{\rx,m},z_{\rx,m})$ denote the positions of the $n$th \gls{Tx} antenna and the $m$th \gls{Rx} antenna, respectively, while $\kk=2\pi/\lambda$ represents the wave number with $\lambda$ being the wavelength. Additionally, $\ba_{\tx}(\cdot)\in \Cset^{N_{\tx}}$ and $\ba_{\rx}(\cdot)\in \Cset^{N_{\rx}}$ are the \gls{Tx} and \gls{Rx} \gls{NF} array responses, respectively. The location of the $s$th scatterer is given by $\bu_s\in\Rset^3$, while $\hat{c}_s$ represents the channel power gain for the $s$th \gls{nLOS} path, with $S$ denoting the total number of scatterers. The normalized LOS channel matrix can be expressed in terms of the \gls{Tx} and \gls{Rx} array responses as follows:
\begin{subequations}
\begin{align} 
	\bH_{\NF}^\LOS  &= [\ba_{\rx}(\bu_{\tx,1}),\dots,\ba_{\rx}(\bu_{\tx,N_{\tx}})]\\
 &= [\ba_{\tx}(\bu_{\rx,1}),\dots,\ba_{\tx}(\bu_{\rx,N_{\rx}})]^\Trans.
\end{align}
\end{subequations}
The \gls{nLOS} channel model in \eqref{Eq:existingMIMO} is based on the point scattering assumption, which models each scattering object as a secondary point source. Under this assumption, the channel amplitude  $\hat{c}_s$  is proportional to  $\frac{1}{\|\bu_{\rx}-\bu_{s}\|\|\bu_s-\bu_{\tx}\|}$, where  $\bu_{\tx}=(x_\tx, y_\tx, z_\tx)$ and $\bu_{\rx}=(x_\rx, y_\rx, z_\rx)$  represent the centers of the transmit and receive arrays, respectively. This implies that, unless the scatterer is located near to either the \gls{Tx} or the \gls{Rx}, the contribution of point-source scattering remains relatively weak compared to the LOS link, whose channel amplitude  $c_0$  is proportional to $\frac{1}{\|\bu_{\rx}-\bu_{\tx}\|}$  \cite{delbari2026far}.

\subsection{Proposed Mixed-Scattering NF NLOS Channel Model}
\begin{figure}[t]
    \centering
    \includegraphics[width=0.4\textwidth]{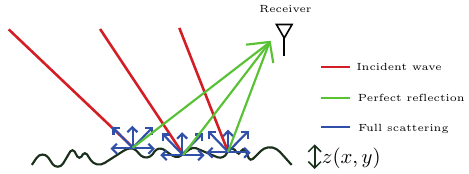}
    \caption{Wave reflected from a rough \gls{NS} comprising reflection along the specular direction as well as scattering.}
    \label{fig:mixed}
    \vspace{-0.5cm}
\end{figure}
In practice, the primary contributors to the \gls{nLOS} components are large reflective surfaces, for example, walls, ground, and ceiling \cite{grossman2017submillimeter}. However, the channel model in \eqref{Eq:existingMIMO} does not accurately capture the impact of such reflectors. Therefore, in this work, we propose a new statistical \gls{NF} \gls{MIMO} channel model that explicitly incorporates the impact of multiple \gls{NS} components. The proposed model is structured as follows:
\begin{equation}
\label{Eq:MIMO_our_model}
	\bH_{\NF} = c_0\bH_{\NF}^\LOS + \sum_{s=1}^S \hat{c}_s\bH^\scat_s + \sum_{r=1}^R c_r\bH^\refl_r,\quad
\end{equation}
where $\bH^\refl_r$ represents the normalized \gls{nLOS} channel matrix associated with the $r$th \gls{NS} in the environment, while $c_r$ is the corresponding channel power gain. The characteristics of the \gls{NS} depend on the operating frequency. When the \gls{NS} is sufficiently smooth, the interaction of an incoming wave with the \gls{NS} leads to ideal specular reflection, whereas a rough surface leads to non-ideal \gls{SS}. In the following section, we focus on modeling $c_r\bH^\refl_r$ in~\eqref{Eq:MIMO_our_model}.
\section{Channel Model for Non-ideal Surface Reflection}
\label{sec: Channel Model for Non-ideal Surface Reflection}
In this section, we characterize the statistical properties of $c_r\bH_r^\refl$. Firstly, we study the physics behind reflections from \glspl{NS}. Then, we decompose $c_r\bH_r^\refl$ into two main components, namely a deterministic and a stochastic component, and characterize each one separately. Finally, we summarize the proposed model and discuss its implications for \gls{NF} \gls{MIMO} system design, including beam training and channel estimation.
\subsection{Underlying Physics of Non-ideal Reflection}
Several factors influence the reflections from an \gls{NS}. Specifically, the surface material \cite{grossman2017submillimeter}, frequency, \gls{AoA} and \gls{AoD} \cite{jansen2011diffuse}, and surface roughness \cite{sheikh2021scattering} lead to scattering, diffraction, and absorption, which results in non-ideal reflection \cite{khamse2023scattering}. Among these, surface roughness relative to the wavelength plays a crucial role in determining the balance between reflection and scattering. Therefore, in this section, we first model the roughness of \glspl{NS} and then examine its impact on the \gls{nLOS} \gls{MIMO} channel matrix.
\subsubsection{Modeling of Roughness}
\ifarxiv
Consider an \gls{NS} lying in the $\x-\y$ plane, with its height fluctuating along the $\z$-axis, as illustrated in Fig.~\ref{fig:mixed}.
\fi
The height at each point on the \gls{NS} can be described as a \gls{RV} following a specific statistical distribution \cite{alissa2019wave}. Similar to common modeling approaches in the literature \cite{jansen2011diffuse,khamse2023scattering,shi2017recovery}, we assume a Gaussian distribution for this \gls{RV}, denoted as $Z$, i.e., $Z\sim\Nset(0,\sigma_z^2)$, where $\sigma_z^2$ is the variance.
\subsubsection{Impact of Roughness on Reflection}
We adopt the \gls{HF} principle to assess how \gls{NS} roughness affects reflection. This principle states that every point along a propagating wavefront serves as a source of secondary spherical waves. Therefore, the new wavefront at any given location emerges from the superposition of these secondary waves \cite{goodman2005introduction,ajam2021channel}. Hence, each term \( c_r[\bH_r^\refl]_{m,n} \) in~\eqref{Eq:MIMO_our_model} can be expressed as follows:
\begin{align}
\label{eq: Huygen}
\frac{E(\bu_{\rx,m})}{E({\bu_{\tx,n}})}=\frac{\zeta}{\jj\lambda}\iint\limits_{\bu\in\Uset} &\underbrace{\frac{z_{\tx,n}}{\|\bu-\bu_{\tx,n}\|^2}\frac{z_{\rx,m}}{\|\bu_{\rx,m}-\bu\|^2}}_{\text{amplitude variations}}\nonumber\\
&\times\underbrace{\e^{\jj \kappa\|\bu-\bu_{\tx,n}\|}\e^{\jj \kappa\|\bu_{\rx,m}-\bu\|}}_{\text{phase variations}}{\dd}A.
\end{align}
Here, $\zeta$ is a constant that ensures surface passivity while also accounting for potential losses, $\bu = (x, y, z)$ represents an arbitrary point on the \gls{NS}, with the \gls{NS} centered at $\bu_c=(0,0,0)$. The term $\dd A$ denotes the area element, while \( E(\bu_{\rx,m}) \) and \( E(\bu_{\tx,n}) \) correspond to the electric fields at the $m$th \gls{Rx} and $n$th \gls{Tx} antennas, respectively. Additionally, the set $\Uset$ represents the total area of the \gls{NS}. Although \eqref{eq: Huygen} accurately captures the impact of reflection, solving it for a specific realization of $z(x,y)$\footnote{$z(x,y),\forall x,y$ is a sample value of the Gaussian random variable $Z$.} across the \gls{NS} does not yield a tractable model for system design and offers no analytical insights. We refer to \eqref{eq: Huygen} as the \gls{HF} integral, and numerically evaluate it under various scenarios to validate the accuracy of the proposed analytical model in Section~\ref{sec: Performance Evaluation}.
\subsection{Proposed NF MIMO Channel Model}
\label{Proposed MIMO NF Channel Model}
We begin by decomposing channel matrix $c_r\bH^\refl_r$ into the following two components:
\begin{equation}\label{Eq:newMIMO}
	c_r\bH^\refl_r = \bar{c}_{r}\bar{\bH}_{r}+\tilde{c}_{r}\tilde{\bH}_{r},
\end{equation}
where $\bar{c}_{r}\bar{\bH}_{r}$ is deterministic and represents the channel mean
\begin{equation}
\label{eq: expectation_H}
    \bar{c}_{r}\bar{\bH}_{r}\triangleq\Ex\{c_r\bH^\refl_r\},
\end{equation}
while $\tilde{c}_{r}\tilde{\bH}_{r}$ accounts for the channel variations and is thus stochastic.
In the following, we characterize the deterministic and stochastic components of $c_r\bH^\refl_r$ separately.
\subsubsection{Deterministic Component}
To simplify the \gls{HF} integral and gain analytical insights into the deterministic component, we apply some reasonable approximations. Specifically, we approximate the area element, amplitude, and phase terms in \eqref{eq: Huygen}. These approximations, commonly used in the literature, enable tractable solutions of complicated \gls{HF} integral equations \cite{goodman2005introduction}.

\textbf{Area Element Approximation:} The area element ($\dd A$) in \eqref{eq: Huygen} is given by the magnitude of the cross product of the tangent vectors \gls{wrt} $\dd x\dd y$ as $|\frac{\partial z}{\partial x}\times\frac{\partial z}{\partial y}|$, where $z$ is a function of $x$ and $y$. This can be written as
\begin{equation}
    \dd A=\sqrt{1+\left(\frac{\partial z}{\partial x}\right)^2+\left(\frac{\partial z}{\partial y}\right)^2}\dd x\dd y.
\end{equation}
First, consider an ideally smooth surface, i.e., $\frac{\partial z}{\partial x}=\frac{\partial z}{\partial y}=0,\,\forall \bu\in\Uset$ (or equivalently $\sigma_z=0$), where the area element in \eqref{eq: Huygen} simplifies to $\dd A = \dd x \dd y$. As the \gls{NS} roughness $\sigma_z$ increases, each point $(x,y)$ on the \gls{NS} continues to act as a secondary source; however, the surface height variations $z(x,y)$ introduce an additional phase shift. In this scenario, we approximate the area element as \( \dd{A} \approx \dd{x} \dd{y} \) in \eqref{eq: Huygen} ($\frac{\partial z}{\partial x}=\frac{\partial z}{\partial y}\approx0$), while accounting for the effect of \gls{NS} roughness in the phase terms \( \e^{\jj \kappa\|\bu-\bu_{\tx,n}\|} \e^{\jj \kappa\|\bu_{\rx,m}-\bu\|} \), which depend on \( \bu \) or, equivalently, on  $z(x,y)$ \cite{shi2017recovery}. This approximation holds when $\sigma_z$ is on the order of the wavelength, which is the regime of interest investigated in this paper.\footnote{We note that this approximation is not valid for arbitrary surface geometries when $\frac{\partial z}{\partial x}\gg1$ and/or $\frac{\partial z}{\partial y}\gg1$.}

\textbf{Amplitude Approximation:} Using a Taylor series expansion, we approximate the inverse squared distance term as follows:
\begin{equation}
\frac{z_{\tx,n}}{\|\bu-\bu_{\tx,n}\|^2}=\frac{z_\tx}{u_\tx^2}+o\left(\frac{u_\tx^\max+u_\NS^\max}{u_\tx}\right),\,\forall \bu\in\Uset, \forall n,
\end{equation}
where $u_\tx = \|\bu_\tx\|$ represents the distance from the center of \gls{Tx} to the origin, and $u_\tx^\max$ and $u_\NS^\max$ denote the maximum dimensions of the \gls{Tx} and the \gls{NS}, respectively. Therefore, the approximation $\frac{z_{\tx,n}}{\|\bu-\bu_{\tx,n}\|^2}\approx\frac{z_\tx}{u_\tx^2}$ holds under the condition:
\begin{equation}
\label{eq: condition approximation}
\frac{u_\tx^\max+u_\NS^\max}{u_\tx} \ll 1,
\end{equation}
which ensures that the variations in $\bu$ over the \gls{NS} and of $\bu_{\tx,n}$ over \gls{Tx} are small relative to the distance of the \gls{Tx} center to the \gls{NS}. Similarly, applying the same approach at the \gls{Rx}, we obtain the approximation:
\begin{equation}
\frac{z_{\rx,m}}{|\bu-\bu_{\rx,m}|^2} \approx \frac{z_\rx}{u_\rx^2},\,\forall \bu\in\Uset, \forall m,
\end{equation}
where $u_\rx = \|\bu_\rx\|$ is the distance from the center of \gls{Rx} to the origin. Substituting these approximations into the \gls{HF} integral \eqref{eq: Huygen}, the channel coefficient for the reflected component is given by
\begin{equation}\label{eq: Huygen amplitude}
c_r[\bH_r^\refl]_{m,n}=c_I\iint\limits_{\bu\in\Uset} I(\bu,\bu_{\tx,n},\bu_{\rx,m}){\dd}x{\dd}y,
\end{equation}
 where $I(\bu,\bu_{\tx,n},\bu_{\rx,m})\triangleq\e^{\jj \kappa\|\bu-\bu_{\tx,n}\|}\e^{\jj \kappa\|\bu_{\rx,m}-\bu\|}$ and $c_I\triangleq\frac{\zeta z_\tx z_\rx}{\jj\lambda u_\tx^2 u_\rx^2}$.
 
 \textbf{Phase Approximation:} The roughness of the \gls{NS} affects the phase term, defined as \( \varphi(\bu) \triangleq \kappa(\| \bu_{\tx,n}-\bu\|+\| \bu_{\rx,m}-\bu\|) \), in (\ref{eq: Huygen amplitude}). By introducing $\bu_\rho\defeq(x,y,0)$, $\varphi(\bu)$ can be simplified as follows:
\begin{align}
    \varphi(\bu)=&\kappa \|\bu_{\tx,n}-\bu_\rho\|\sqrt{1-\frac{2z_{\tx,n} z+z^2}{\|\bu_{\tx,n}-\bu_\rho\|^2}}\nonumber\\
    &+\kappa\|\bu_{\rx,m}-\bu_\rho\|\sqrt{1-\frac{2z_{\rx,m} z+z^2}{\|\bu_{\rx,m}-\bu_\rho\|^2}}\nonumber
\end{align}
\begin{equation}
\label{eq: orders}
\overset{(a)}{=}\!\! \kappa\big(\| \bu_{\tx,n}-\bu_\rho\|+\| \bu_{\rx,m}-\bu_\rho\|\big)- \kappa_zz+\bigo(z^2),
\end{equation}
where we used the Taylor approximation $\sqrt{1-2w}\approx1-w$ for small $w$, and $\kappa_z=\kappa\left(\frac{z_{\tx,n}}{\|\bu_{\tx,n}-\bu\|}+\frac{z_{\rx,m}}{\|\bu_{\rx,m}-\bu\|}\right)$. We can approximate $\kappa z\frac{z_{\tx,n}}{\|\bu_{\tx,n}-\bu\|}\approx\kappa z\frac{z_\tx}{u_\tx}$ and $\kappa z\frac{z_{\rx,m}}{\|\bu_{\rx,m}-\bu\|}\approx\kappa z\frac{z_\rx}{u_\rx}$ under condition \eqref{eq: condition approximation} for both \gls{Tx} and \gls{Rx}. With these approximations, we can rewrite $\kappa_z=\kappa(\cos(\theta_{\tx})+\cos(\theta_{\rx}))$ where $\cos(\theta_{\tx})=\frac{z_\tx}{u_\tx}$, $\cos(\theta_{\rx})=\frac{z_\rx}{u_\rx}$. It is finally noted that equality $(a)$ follows from the fact that $\frac{z}{\|\bu_{\tx,n}-\bu_\rho\|}$ and $\frac{z}{\|\bu_{\rx,m}-\bu_\rho\|}$ are typically small and can be neglected.

With the approximations of the area element, amplitude, and phase in hand, we can now compute the expected value of $c_r[\bH_r^\refl]_{m,n}$ in (\ref{eq: Huygen amplitude}) \gls{wrt} the \gls{RV} $Z$ \cite{shi2017recovery} by omitting $\bigo(z^2)$, 
and obtain the following expression using (\ref{eq: expectation_H}):
\begin{align}
&\bar{c}_{r}[\bar{\bH}_{r}]_{m,n}=\Ex\{c_r[\bH_r^\refl]_{m,n}\}\nonumber\approx c_I\\
\label{eq: field expected}
     &\times\kern-0.5em\iint\limits_{\bu\in\Uset}\Ex\{\e^{-\jj \kappa_z z}\}\underbrace{\e^{\jj \kappa(\| \bu_{\tx,n}-\bu_\rho\|+\| \bu_{\rx,m}-\bu_\rho\|)}}_{\kern-0em \text{Deterministic and unaffected by $z$}}\dd x\dd y.
\end{align}
\begin{lem}
\label{lem: Gaussian distribution}
    If \gls{RV} $Z$ follows a zero-mean Gaussian distribution with variance $\sigma_z^2$, i.e., $Z\sim\Nset(0,\sigma_z^2)$, then the expected value of $\e^{-\jj \kappa_z z}$, where $\kappa_z$ is a constant, is given as follows:
    \begin{equation}
        \Ex\{\e^{-\jj \kappa_z z}\}=\e^{-\frac{\sigma_z^2\kappa_z^2}{2}}.
    \end{equation}
\end{lem}
\begin{proof}
We have
    \begin{align}
        &\Ex\{\e^{-\jj \kappa_z z}\}\!=\!\!\int_{-\infty}^\infty \!\e^{-\jj \kappa_z z} \frac{1}{\sqrt{2\pi\sigma_z^2}}\e^{-(\frac{z^2}{2\sigma_z^2})} \dd z=\e^{-\frac{\sigma_z^2\kappa_z^2}{2}}\nonumber\\
        &\times\!\!\underbrace{\int_{-\infty}^\infty \frac{1}{\sqrt{2\pi\sigma_z^2}}\e^{-(\frac{z}{\sqrt{2}\sigma_z}+\jj\frac{\sqrt{2}\sigma_z\kappa_z}{2})^2} \!\dd z}_{=1}\!=\!\e^{-\frac{\sigma_z^2\kappa_z^2}{2}}\!\!.
    \end{align}
    This completes the proof.
\end{proof}
\begin{theorem}
\label{Theorem Gaussian distribution}
    Consider an \gls{NS} that is extremely large compared to the \gls{Tx}/\gls{Rx} apertures, and assume that the height variations of the \gls{NS} follow the conditions in Lemma~\ref{lem: Gaussian distribution}. Then, we have the following expressions:
\begin{subequations}
\label{eq:coherent component}
    \begin{align}
&\bar{c}_{r}(g)&=\frac{\zeta_r}{\jj\lambda}\frac{\e^{-\frac{g}{2}}}{\|\bu^r_{\vrx}-\bu_{\tx}\|}&=\frac{\zeta_r}{\jj\lambda}\frac{\e^{-\frac{g}{2}}}{\|\bu_{\rx}-\bu^r_{\vtx}\|},\label{eq:coherent component a}\\
\big[&\bar{\bH}_{r}\big]_{m,n} &= \e^{\jj\kk\|\bu^r_{\vrx,m}-\bu_{\tx,n}\|} &= \e^{\jj\kk\|\bu_{\rx,m}-\bu^r_{\vtx,n}\|},\label{eq:coherent component b}
\end{align}
 \end{subequations}
where $\bu^r_{\vtx,n}$ and $\bu^r_{\vrx,m}$ denote the virtual images of the $n$th \gls{Tx} antenna and the $m$th \gls{Rx} antenna mirrored on the $r$th \gls{NS}, $\bu^r_{\vtx}$ and $\bu^r_{\vrx}$ are the centers of the mirror images of the Tx and Rx arrays, respectively, 
and
\ifarxiv
\begin{equation}
g=(\kk_z\sigma_z)^2.
\end{equation}
\fi
\end{theorem}
\begin{proof}
    \label{app: Gaussian distribution}
By substituting $g=(\kappa_z\sigma_z)^2$ into Lemma~\ref{lem: Gaussian distribution}, we can rewrite \eqref{eq: field expected} as follows:
    \begin{equation}
    \label{eq: integral perfect reflection}
        \bar{c}_{r}[\bar{\bH}_{r}]_{m,n}=\e^{-\frac{g}{2}}\underbrace{c_I\iint\limits_{\bu\in\Uset}\e^{\jj \kappa(\| \bu_{\tx,n}-\bu_\rho\|+\| \bu_{\rx,m}-\bu_\rho\|)}{\dd}x{\dd}y}_{\text{Perfect reflection}}.
    \end{equation}
Since the integral term in \eqref{eq: integral perfect reflection} is independent of $z$, it has the same value as in the case of perfect reflection, i.e., for $z(x,y)=0,\,\,\forall x,y\in\Uset$. Consequently, the deterministic component of the channel, $[\bar{\bH}_{r}]_{m,n}$, can be expressed as in \eqref{eq:coherent component} derived using image theory and geometric optics \cite{a2005antenna}. This concludes the proof.
\end{proof}
Theorem~\ref{Theorem Gaussian distribution} establishes that the mean of $c_r\bH_r^\refl$ retains the structure of ideal specular reflection, while its amplitude, $\bar{c}_{r}$, decreases as the \gls{NS} roughness increases, i.e., as $\sigma_z$ grows. This result will be validated in Section~\ref{Simulation Result} by comparing it with the \gls{HF} integral in \eqref{eq: Huygen} for practically large \glspl{NS}.
\subsubsection{Stochastic Component}\label{Stochastic Component}
Unlike the point scattering channel model in \eqref{Eq:existingMIMO}, the channel matrix for a rough \gls{NS} is not fully deterministic. Instead, it includes statistical variations due to the surface roughness. In the following, we derive a statistical model for the \gls{NF} channel to capture these effects by considering two general scenarios: One without surface correlation, which offers analytical simplicity and valuable insights, and one with surface correlation, which is more representative of practical environments. For each case, we first derive the joint \gls{PDF} of the elements of $\tilde{\bH}_{r}$, followed by an analysis of the covariance matrix of the channel elements and the resulting channel power gain.

\textit{2A) Without surface correlation:} In this case, we assume that the points on the surface of the \gls{NS} are statistically independent, i.e., no surface correlation exists between different points on the surface. Although this assumption is not strictly valid, it constitutes a good approximation for very rough surfaces, and it provides insights with tractable analysis\footnote{We note that assuming zero-length correlation implies that values at infinitesimally close points are uncorrelated. As a result, the quantity over an infinitesimal neighborhood behaves like an independent variable, potentially leading to large variations in slope. This contradicts the area element approximation $\dd A = \dd x\dd y$, which relies on smoothness over small regions. Nonetheless, we address and resolve this issue in Section~\ref{Stochastic Component}\textcolor{blue}{B} by considering non-zero surface correlation.}.

\textbf{Joint PDF of the Elements of $\tilde{\bH}_{r}$ in \eqref{Eq:newMIMO}:}
 $\varphi(\bu)$ in (\ref{eq: orders}) follows a certain distribution that is identical for all $\bu \in \Uset$. Consequently, the channel response in (\ref{eq: Huygen}) can be viewed as a sum of a large number of \gls{iid} random variables. According to the multidimensional \gls{CLT} \cite{papoulis2002probability}, this implies that the elements of $\tilde{\bH}_{r}$ follow a joint Gaussian distribution. This observation will be validated through simulations in Section~\ref{Simulation Result}. For a joint Gaussian distribution, the first- and second-order moments fully characterize the statistics. Since $\tilde{\bH}_{r}$ has zero mean by definition, we next focus on analyzing its covariance matrix.

\textbf{Covariance matrix:} The covariance of channel coefficients $c_r[\bH_r^\refl]_{n,m}$ and $c_r[\bH_r^\refl]_{n',m'}$ is denoted as $\Cov\{c_r[\bH_r^\refl]_{n,m},c_r^*[\bH_r^\refl]_{n',m'}^*\}$. Using Theorem~\ref{Theorem Gaussian distribution} and the notation $\alpha_{n,m}(g)\triangleq \bar{c}_{r}(g)[\bar{\bH}_{r}]_{n,m}$ for brevity, yields
\begin{align}
\label{eq:correlation1}
    &\Cov\{c_r[\bH_r^\refl]_{n,m},c_r^*[\bH_r^\refl]_{n',m'}^*\}\nonumber\\&=\Ex\left\{(c_r[\bH_r^\refl]_{n,m}-\alpha_{n,m}(g))(c_r[\bH_r^\refl]_{n',m'}-\alpha_{n',m'}(g))^*\right\}\nonumber\\
    &\overset{(a)}{=}\!\Ex\Big\{|c_I|^2\!\!\iint\limits_{\bu\in\Uset}\! I(\bu,\bu_{\tx,n},\bu_{\rx,m})\dd x\dd y\nonumber\\
    &\quad\quad\times\iint\limits_{\bu'\in\Uset}\! I^*(\bu'\!,\bu_{\tx,n'}\!,\bu_{\rx,m'}\!)\dd x'\dd y'\Big\}-|\bar{c}_{r}(g)|^2,
\end{align}
where $(a)$ follows from \eqref{eq: expectation_H} and \eqref{eq: Huygen amplitude}. In general, further simplification of \eqref{eq:correlation1} requires additional assumptions. To facilitate analysis, we categorize the surface roughness $\sigma_z$ into three regimes:\\
\textbf{Regime 1: $\kappa\sigma_z\ll1$}: In this case, the \gls{NS} is sufficiently smooth leading to perfect \gls{SR}.\\
\textbf{Regime 2:} This is a transient regime between Regimes 1 and 3, where both surface reflection and scattering are present.\\
\textbf{Regime 3: $\kappa\sigma_z\gg1$}: In this case, the \gls{NS} scatters the wave in all directions leading to full \gls{SS}.\\

In Regime 1, the channel is completely deterministic and can be directly obtained from Theorem~\ref{Theorem Gaussian distribution}. In contrast, for Regime 2, the \gls{HF} integral in \eqref{eq: Huygen} does not lead to a closed-form solution and must be evaluated numerically. Meanwhile, in Regime 3, the deterministic channel coefficient approaches zero as $|\bar{c}_{r}(g\to+\infty)|\to0$. To characterize the spatial correlation, we define the normalized covariance as
\begin{equation}
\label{eq: spatial correlation}
    [\bR]_{(n,m,n',m')}\triangleq\frac{1}{\Ex\{|c_r|^2\}}\Cov\{c_r[\bH_r^\refl]_{n,m},c_r^*[\bH_r^\refl]_{n',m'}^*\},
\end{equation}
and we refer to this quantity as the spatial correlation \cite{Arbitrary_corr,Matthaiou2012}. We further assume that:

Assumption 1: $\frac{2\|\bu_{\tx,n}-\bu_{\tx,n'}\|^2}{\lambda}<u_\tx$ and $\frac{2\|\bu_{\rx,m}-\bu_{\rx,m'}\|^2}{\lambda}<u_\rx$ hold.\\
Note that Assumption 1 holds in practice, as the distance between closely spaced antennas with significant correlation is typically much smaller than the \gls{Tx}-\gls{Rx} distance. For ease of analysis, we define a local spherical coordinate system at the \gls{Rx}, where the $\z$-axis aligns with the \gls{Rx} antenna positions $\bu_{\rx,m}$ and $\bu_{\rx,m'}$, and the origin is placed at their midpoint, $\frac{\bu_{\rx,m} + \bu_{\rx,m'}}{2}$. Similarly, a corresponding local coordinate system is considered at the \gls{Tx} \gls{wrt} \gls{Tx} antenna locations $\bu_{\tx,n}$ and $\bu_{\tx,n'}$. Using these coordinate definitions and Assumption 1, we obtain the following lemma.
\begin{lem}
    \label{lem: tworxtx}
     For Regime 3 and under Assumption 1, (\ref{eq:correlation1}) can be simplified as follows:
    \begin{align}
        \label{eq: lemma 1}
        [\bR]_{(n,m,n',m')}
        =\frac{1}{|\Uset|}\iint\limits_{\bu\in\Uset} &\e^{\jj \kappa(\|\bu_{\rx,m}-\bu_{\rx,m'}\|\sin(\theta_\rx^l(\bu)))}\\
        \times&\e^{\jj \kappa(\|\bu_{\tx,n}-\bu_{\tx,n'}\|\sin(\theta_\tx^l(\bu)))}\dd x\dd y,\nonumber
    \end{align}
    where $\theta_\rx^l(\bu)$ and $\theta_\tx^l(\bu)$ are the elevation angles 
    of point $\bu$ evaluated
    in the local \gls{Rx} and \gls{Tx} coordinate systems, respectively.
\end{lem}
\begin{IEEEproof}
    \ifarxiv
    The proof is provided in Appendix~\ref{app: tworxtx}.
    \fi
\end{IEEEproof}
To gain insight from Lemma~\ref{lem: tworxtx}, we focus on a single \gls{Tx} antenna and calculate the spatial correlation between two \gls{Rx} antennas located at $\bu_{\rx,m}$ and $\bu_{\rx,m'}$, denoted by $[\bR]_{m,m'}$. Due to channel reciprocity, the correlation from one \gls{Rx} antenna to two \gls{Tx} antennas exhibits a similar behavior. We assume that the elevation angles $\theta(\bu)$ for all points $\bu \in \Uset$ on the \gls{NS} lie within the interval $[\theta_1, \theta_2]$, when expressed in the local spherical coordinate system at the \gls{Rx}.

\begin{figure}[t]
\centering
\begin{subfigure}{0.24\textwidth}
\includegraphics[width=\textwidth]{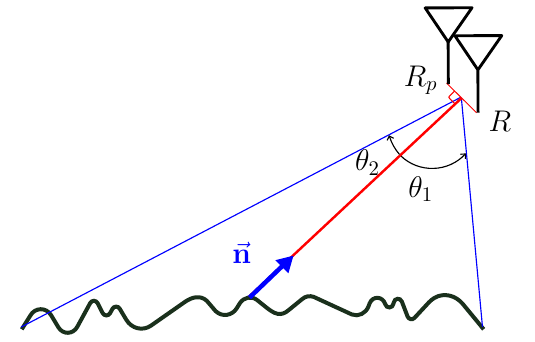}
    \vspace{-0.5cm}
\end{subfigure}
\begin{subfigure}{0.24\textwidth}
\includegraphics[width=\textwidth]{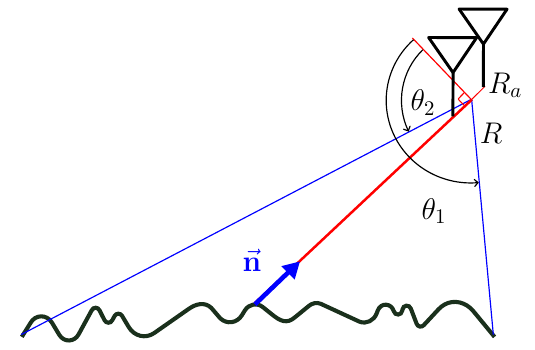}
    \vspace{-0.5cm}
\end{subfigure}
\caption{Vector $\vec \bn$ is aligned with $\bu_R-\bu_{R_a}$ and perpendicular to $\bu_R-\bu_{R_p}$.}
\vspace{-0.5cm}
\label{fig:mixed2}
\end{figure}

\begin{prop}
\label{prop: sinc}
    Assuming isotropic scattering within $\theta_1<\theta(\bu)<\theta_2$, the elements of the spatial correlation matrix $\bR$ are given by
    \begin{equation}
    \label{eq: sinc 1}
    |[\bR]_{m,m'}|=\sinc\Big(\frac{2d_{m,m'}}{\lambda}\cos\left(\frac{\theta_2+\theta_1}{2}\right)\sin\left(\frac{\theta_2-\theta_1}{2}\right)\Big),
\end{equation}
where $d_{m,m'}=\|\bu_{\rx,m}-\bu_{\rx,m'}\|$ and $\sinc(x)=\frac{\sin(\pi x)}{\pi x}$.
\end{prop}

\begin{IEEEproof}
\ifarxiv
    The proof is provided in Appendix~\ref{app: sinc}.
\fi
\end{IEEEproof}
Proposition~\ref{prop: sinc} shows that the spatial correlation is maximized (minimized) when vector $\bu_m - \bu_{m'}$ is parallel (perpendicular) to vector $\vec{\bn}$, which denotes the direction from the center of the \gls{NS} to the center of the \gls{Rx}. To visualize this, consider a reference antenna $R$ and two other antennas, $R_a$ (aligned) and $R_p$ (perpendicular), placed such that the vector from $R$ to $R_a$ is parallel to $\vec{\bn}$, and the vector from $R$ to $R_p$ is perpendicular to $\vec{\bn}$. These placements are illustrated in Fig.~\ref{fig:mixed2}. In the following corollary, we simplify \eqref{eq: sinc 1} for these two special cases.
\begin{corol}
Under the assumptions of Proposition~\ref{prop: sinc}, we have $|[\bR]_{m,m'}|$
\label{corol: sinc}
$$
=
\begin{cases}
\sinc\Big(\frac{2\|\bu_R-\bu_{R_a}\|}{\lambda}\sin^2\left(\frac{\theta_c}{2}\right)\Big), \theta_1=-\frac{\pi}{2},\, \theta_2=-\frac{\pi}{2}+\theta_c,\\
\sinc\Big(\frac{2\|\bu_R-\bu_{R_p}\|}{\lambda}\sin\left(\frac{\theta_c}{2}\right)\Big),\,\, \theta_1=-\frac{\theta_c}{2},\, \theta_2=\frac{\theta_c}{2},
\end{cases}
$$
where we have used the notation $\theta_c\triangleq\theta_2-\theta_1$.
\end{corol}
\begin{IEEEproof}
    The proof can be obtained by substituting the definitions of $\theta_1$ and $\theta_2$ into \eqref{eq: sinc 1} and using further trigonometric manipulations to simplify the expression.
\end{IEEEproof}
Corollary~\ref{corol: sinc} yields the following insights. First, since $\sin^2(x) < \sin(x),\quad \forall x \in (0, \pi)$, the spatial correlation between antennas is higher when $\bu_m - \bu_{m'}$ is parallel to $\vec{\bn}$, compared to the perpendicular case, assuming equal distances to the reference points, i.e., $|\bu_R - \bu_{R_a}| = |\bu_R - \bu_{R_p}|$. Second, increasing the angular spread $\theta_c$ of the \gls{AoA} reduces the correlation. This occurs when the \gls{NS} is larger or the \gls{Rx} antennas are located closer to the \gls{NS}. The accuracy of the expressions in Corollary~\ref{corol: sinc} will be validated through simulations in Section~\ref{Simulation Result}.

\textbf{Channel power gain:} In this part, we focus on the situation with only one \gls{Tx} and one \gls{Rx}. Based on Theorem~\ref{Theorem Gaussian distribution}, the channel power gain in Regime 1 can be directly calculated, as the channel is deterministic. For Regime 3, the channel power gain can be derived analytically by exploiting the law of conservation of energy and the symmetry inherent to full scattering. In this regime, the deterministic component vanishes, i.e., $\e^{-\frac{g}{2}} \to 0$ in (\ref{eq:coherent component}), and the reflected channel power is argued to be uniformly distributed across all directions due to the homogeneous nature of the \gls{SS}. We define the power gain of the reflected channel from the $r$th \gls{NS} in this regime as $\tilde{c}_{r}(+\infty)$, which can be expressed as follows
\begin{equation}
\label{eq: channel power gain without surface correlation}
    |\tilde{c}_{r,+\infty}|^2\triangleq\frac{P_\rx}{P_\tx}=\frac{\zeta_r}{\lambda}\Big(\frac{A_\rx D_r}{4\pi u_\rx^2}\Big)\Big(\frac{A_r D_\tx}{4\pi u_\tx^2}\Big),
\end{equation}
where $P_\rx$ and $P_\tx$ denote the \gls{Rx} and \gls{Tx} power, respectively, $A_\rx$ and $A_r$ are the effective areas of the \gls{Rx} and the \gls{NS}, and $D_\tx$ and $D_r$ represent the directivities of the \gls{Tx} and the \gls{NS}, respectively. We assume $D_r = 2$ (i.e., 3~dB gain), as the \gls{NS} reflects energy uniformly over half of the space. For Regime 2, unfortunately, a closed-form solution does not exist, and (\ref{eq: Huygen}) must be solved directly. Since this is analytically intractable, we adopt the following model for the channel power gain of the stochastic component:
\begin{equation}
\label{eq: stochastic power}
    \Ex\{|\tilde{c}_{r}(g)|^2\}=(1-\e^{-\frac{g}{2}})^2|\tilde{c}_{r,+\infty}|^2,
\end{equation}
which aligns with the limiting cases for Regimes 1 ($g = 0$) and 3 ($g \to +\infty$). Based on \eqref{eq: stochastic power}, we also propose the following expression for the total channel power gain:
\begin{equation}
\label{eq: total power}
    \Ex\{|c_r|^2\}=\bar{c}_{r}^2(0)\e^{-g}+(1-\e^{-\frac{g}{2}})^2|\tilde{c}_{r,+\infty}|^2.
\end{equation}

In Section~\ref{Simulation Result}, we will verify via simulation (see Fig.~\ref{fig:3regimes}) that this heuristic model closely approximates the results obtained through numerical evaluation of the \gls{HF} integral in \eqref{eq: Huygen}.

\textit{2B) With surface correlation:} In this case, we assume a local dependence exists among different points on the \gls{NS}, introducing surface correlation. Accounting for this correlation enhances the accuracy of the channel model.

\textbf{Joint PDF of the elements of $\tilde{\bH}_{r}$ in \eqref{Eq:newMIMO}:} 
Due to surface correlation, the heights $z(x,y)$ are no longer \gls{iid}. However, if the aperture of the \gls{NS} is much larger than the correlation length, the surface can be approximated as comprising many locally correlated but mutually independent segments. The sum of these independent segments converges to a joint Gaussian random variable by the multidimensional \gls{CLT}, so the \gls{PDF} of the channel is still Gaussian.
Nevertheless, the surface correlation may impact the channel covariance matrix and the total power as discussed below.

\textbf{Covariance matrix:} We start from \eqref{eq:correlation1} to derive the covariance matrix. This time, we cannot decompose the integral into two parts for $\bu\neq\bu'$ and $\bu=\bu'$, as we did to prove Lemma~\ref{lem: tworxtx}, due to the surface correlation. To evaluate \eqref{eq:correlation1}, we begin by interchanging the expectation and the integrals. We then compute the expectation using the phase expression introduced in \eqref{eq: orders}, as follows:
\begin{align}
    &\Ex\Big\{\!I(\bu,\bu_{\tx,n},\bu_{\rx,m})I^*(\bu'\!,\bu_{\tx,n'}\!,\bu_{\rx,m'}\!)\Big\}
    \!\!\!\overset{(a)}{=}\!\!\Ex\Big\{\!\e^{\jj\kappa F+\jj\kappa_z(z-z')}\!\!\Big\},\label{eq: correlation with z and z'}
\end{align}
where $F\defeq\| \bu_{\tx,n}-\bu_\rho\|-\| \bu_{\tx,n'}-\bu'_\rho\|+\| \bu_{\rx,m}-\bu_\rho\|-\| \bu_{\rx,m'}-\bu'_\rho\|$ and $(a)$ follows from \eqref{eq: orders} by omitting $\bigo(z^2)$ and $\bigo(z'^2)$ terms. The only \glspl{RV} in \eqref{eq: correlation with z and z'} are $z$ and $z'$, which follow a joint Gaussian distribution with a general correlation coefficient $C(\rho)$, where $\rho=\|\bu'-\bu\|$ denotes radial distance, see Fig.~\ref{fig: correlated rough surface}. Therefore, the expectation operator only applies to the term $\e^{\jj\kappa_z(z-z')}$ resulting in 
\begin{align}
    \Ex\{\e^{\jj\kappa_z(z-z')}\}&=\!\!\int\limits_{\!\!\!\!\!\!\!\!-\infty}^\infty\int\limits_{\!\!\!\!\!-\infty}^\infty \! \frac{\e^{\jj\kappa_z(z-z')}}{2\pi\sigma_z^2\sqrt{1-C(\rho)^2}}\nonumber\\
    &\quad\quad\times\e^{-\frac{1}{2\sigma_z^2(1-C(\rho)^2)}(z^2-2C(\rho) zz'+z'^2)} \dd z\dd z'\nonumber\\
    &\overset{(b)}{=}\e^{-\sigma_z^2\kappa_z^2(1-C(\rho))},
    \label{eq: correlation joint 1}
\end{align}
where $(b)$ follows from the closed-form result given in \cite[Eq. 13]{shi2017recovery}. By substituting the result of \eqref{eq: correlation joint 1} into \eqref{eq: correlation with z and z'} and then into \eqref{eq:correlation1}, we obtain
\begin{align}
\label{eq:correlation1 joint}
    &\Cov\{c_r[\bH_r^\refl]_{n,m},c_r^*[\bH_r^\refl]_{n',m'}^*\}\nonumber\\
&=|c_I|^2\times\!\!\iint\limits_{\!\!\!\!\bu\in\Uset}\!\!\iint\limits_{\bu'\in\Uset}\! \e^{-\sigma_z^2\kappa_z^2(1-C(\rho))+\jj\kappa F}\dd x\dd y\dd x'\dd y'-|\bar{c}_{r}(g)|^2.
\end{align}
Now, we expand $\|\bu_{\tx,n'}-\bu'_\rho\|$ in the definition of $F$ \gls{wrt} $\|\bu_{\tx,n'}-\bu_\rho\|$ as follows: 
\begin{align}
    \label{eq: expand ut-u' in ut-u}
    \|\bu_{\tx,n'}-\bu'_\rho\|\overset{(a)}{=}&\|\bu_{\tx,n'}-\bu_\rho\|+\rho\cos(\Psi_\tx)\nonumber\\
    &+\frac{\rho^2\sin^2(\Psi_\tx)}{2\|\bu_{\tx,n'}-\bu_\rho\|}+\bigo(\rho^3),
\end{align}
where $\Psi_\tx$ is the angle between vectors $\bu_{\tx,n'}-\bu_\rho$ and $\bu'_\rho-\bu_\rho$ and $\rho=\|\bu'_\rho-\bu_\rho\|$ (see Fig.~\ref{fig: correlated rough surface}). Moreover, $(a)$ follows from \cite[Lemma~1]{delbari2026far}. A similar expansion can be obtained for the \gls{Rx} side by substituting indices $\tx$ and $n$ with $\rx$ and $m$, respectively. Substituting \eqref{eq: expand ut-u' in ut-u} into the definition of $F$ inside of \eqref{eq:correlation1 joint} and omitting $\bigo(\rho^3)$ yields:
\begin{align}
    \label{eq: F simplification}
    F&\!\!=\!\!\| \bu_{\tx,n}\!-\!\bu_\rho\|\!-\!\| \bu_{\tx,n'}\!-\!\bu_\rho\|\!-\!\rho\cos(\Psi_\tx)\!-\!\frac{\rho^2\sin^2(\Psi_\tx)}{2\|\bu_{\tx,n'}-\bu_\rho\|}\nonumber\\
    +&\| \bu_{\rx,m}\!\!-\!\bu_\rho\|\!-\!\| \bu_{\rx,m'}\!\!-\!\bu_\rho\|\!-\!\rho\cos(\Psi_\rx)\!-\!\frac{\rho^2\sin^2(\Psi_\rx)}{2\|\bu_{\rx,m'}-\bu_\rho\|}.
\end{align}
\begin{figure}
    \centering
    \includegraphics[width=0.25\textwidth]{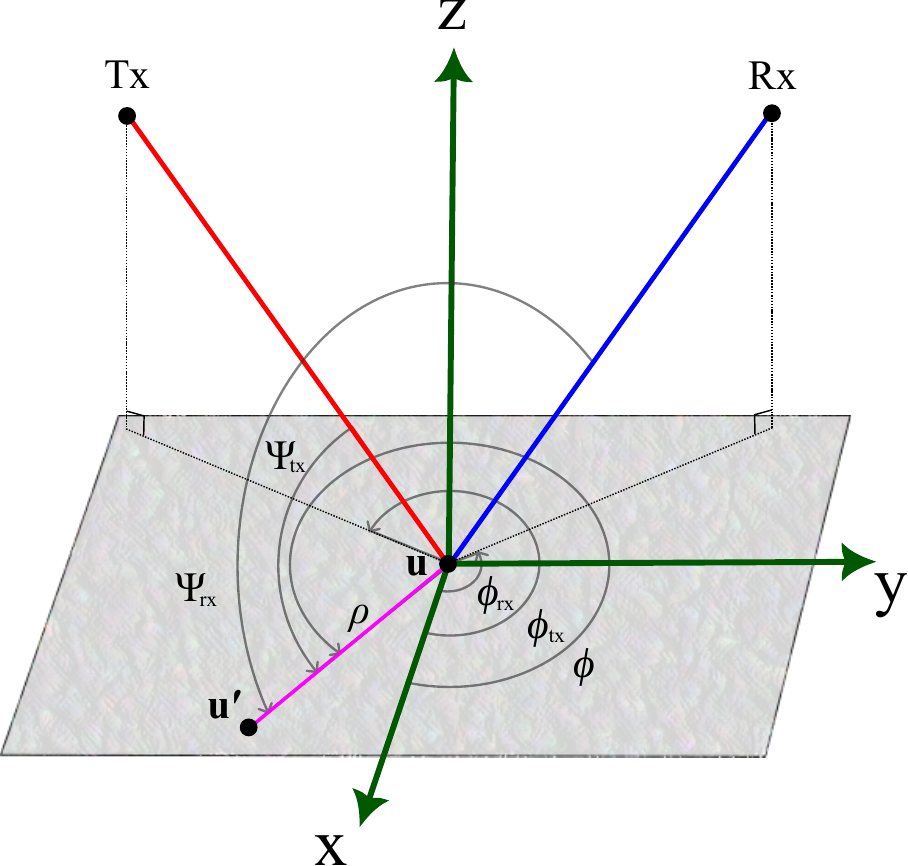}
    \caption{Reflection from a correlated rough surface.}
    \label{fig: correlated rough surface}
    \vspace{-0.5cm}
\end{figure}
To simplify \eqref{eq:correlation1 joint}, we need some assumptions on $C(\rho)$ as follows:

Assumption 2: We assume that when the term $\frac{\kappa\rho^2}{\|\bu_{q_1,q_2'}-\bu_\rho\|}$ is sufficiently large, e.g., larger than $\frac{\pi}{8}$, the correlation factor $C(\rho)$ becomes approximately zero, i.e., $C(\rho) \approx 0$. Here, $q_1=\{\tx,\rx\}$, $q_2=\{n,m\}$, and $\|\bu_{q_1,q_2'}-\bu_\rho\|$ denotes the distance of the \gls{Tx}/\gls{Rx} antenna from \gls{NS}. 

For typical indoor and outdoor scenarios, $\|\bu_{q_1,q_2'}-\bu_{\rho}\|$ is on the order of meters or tens of meters, respectively. This implies that $\rho$ must be larger than $\sqrt{\frac{10\lambda}{16}}$ (approximately 97~mm at $f=20$ GHz and 56~mm at $f=60$ GHz for centimeter and \gls{mmWave} bands, respectively) for $\frac{\kappa\rho^2}{\|\bu_{q_1,q_2'}-\bu_\rho\|}\geq\frac{\pi}{8}$ to hold. This assumption is practically justified, as empirical observations show that $C(\rho)$ becomes negligible when $\rho\geq20-30$~mm \cite[Table I]{ma2019terahertz} or even $\rho\geq1$~mm for the study in \cite{benstock2014influence}. Therefore, the term $\frac{\kappa\rho^2}{|\bu_{q_1,q_2'}-\bu_\rho|}$ has a negligible impact in the integral in \eqref{eq:correlation1 joint}.
In the following lemma, we prove that the spatial correlation, assuming non-zero surface correlation, is also given by \eqref{eq: lemma 1}.
\begin{lem}
    \label{lem: tworxtx correlation}
     By assuming Regime 3 under Assumptions 1 and 2, the spatial correlation (normalized form of \eqref{eq:correlation1 joint}) can be simplified to \eqref{eq: lemma 1}.
\end{lem}
\begin{IEEEproof}
    \ifarxiv
    The proof is provided in Appendix~\ref{app: tworxtx correlation}.
    \fi
\end{IEEEproof}
Lemma~\ref{lem: tworxtx correlation} reveals that the spatial correlation remains the same as in the case without surface correlation for a sufficiently large \gls{NS} under Assumptions 1 and 2. Next, we investigate whether surface correlation impacts the channel power gain.

\textbf{Channel power gain:} In this part, we focus on the situation with only one \gls{Tx} and one \gls{Rx}. The channel power gain in Regime 1 can be determined using Theorem~\ref{Theorem Gaussian distribution}. However, in Regime 3, the channel power gain generally depends on $C(\rho)$. A typical model for $C(\rho)$ is the Gaussian function, i.e., $C(\rho) = e^{-\frac{\rho^2}{\ell^2}}$, where $\ell$ denotes the correlation length of the \gls{NS} \cite{beckmann1987scattering, shi2017recovery}. Under this model, we can approximate $1 - C(\rho) \approx \frac{\rho^2}{\ell^2}$ for small $\frac{\rho}{\ell}$. Note that this approximation becomes invalid when $\ell$ is very small. To address this, we define a threshold $\ell_{\min}$ such that the approximation holds for $\ell>\ell_{\min}$. Using this approximation in \eqref{eq:correlation1 joint}, we can derive the following theorem.
\begin{theorem}
    \label{theorem power regime 3}
    Assuming $1 - C(\rho) \approx \frac{\rho^2}{\ell^2}$, we can simplify \eqref{eq:correlation1 joint} and derive the channel power gain for Regime 3 as follows:
    \begin{equation}
    \label{eq: theorem power regime 3}
        \Ex\{|c_r|^2\}=\frac{P_\rx}{P_\tx}=\frac{B\ell^2}{2\kappa_z^2\sigma_z^2}\e^{-\frac{(\kappa_\rho\ell)^2}{(2\kappa_z\sigma_z)^2}},\ell_{\min}<\ell<\ell_{\max},
    \end{equation}
    where $B$ is a constant, $\kappa_\rho$ is a linear function of $\kappa$ whose slope depends on positions $\bu_\tx$ and $\bu_\rx$, and $\ell_{\max}=2\sigma_z\frac{\kappa_z}{\kappa_\rho}$.
\end{theorem}
\begin{IEEEproof}
    \ifarxiv
    The proof is provided in Appendix~\ref{app: power regime 3}.
    \fi
\end{IEEEproof}
Theorem~\ref{theorem power regime 3} provides a mathematical expression for the channel power gain as a function of the correlation length. To determine the constant $B$, we use the fact that, when $\ell=\ell_{\max}$, the surface becomes effectively smooth. In this regime, the surface correlation length is large enough for the roughness to be neglected, allowing us to apply the model for Regime 1 to derive $B$. Therefore, the power gain should equal $|\bar{c}_r(0)|^2$, as derived in \eqref{eq:coherent component a}. Substituting $\ell=\ell_{\max}$ into \eqref{eq: theorem power regime 3}, yields $|\bar{c}_r(0)|^2=\frac{2B}{\kappa_\rho^2}\e^{-1}$, from which we solve for $B$ as $B=\frac{|\bar{c}_r(0)|^2\kappa_\rho^2}{2}\e$. To find $\ell_{\min}$, we refer to the discussion in Section~\ref{Stochastic Component}\textcolor{blue}{A} regarding the channel power gain using the law of conservation of energy. Specifically, when $\ell = \ell_{\min}$, the power should match $|\tilde{c}_{r,+\infty}|^2$. Substituting into the power expression, we obtain the following expression:
\begin{equation}
\frac{|\bar{c}_r(0)|^2\kappa_\rho^2\ell_{\min}^2}{4\kappa_z^2\sigma_z^2}\e^{1-\frac{(\kappa_\rho\ell_{\min})^2}{(2\kappa_z\sigma_z)^2}}=|\tilde{c}_{r,+\infty}|^2.
\end{equation}
To simplify notation, define $S\defeq\frac{(\kappa_\rho\ell)^2}{(2\kappa_z\sigma_z)^2}$, $S_{\max}\defeq\frac{(\kappa_\rho\ell_{\max})^2}{(2\kappa_z\sigma_z)^2}=1$, and $S_{\min}$ to be the solution of $S_{\min}e^{1-S_{\min}}=\frac{|\tilde{c}_{r,+\infty}|^2}{|\bar{c}_{r}(0)|^2}$. We are now ready to propose the following expression for the channel power gain:
\begin{equation}
\label{eq: S function}
    \frac{P_\rx}{P_\tx}=\begin{cases}
        |\bar{c}_{r}(0)|^2, &\text{ for $S\geq 1$},\\
        |\bar{c}_{r}(0)|^2Se^{1-S}, &\text{ for $S_{\min}< S< 1$},\\
        |c_{n,r,+\infty}|^2, &\text{ for $S\leq S_{\min}$}.
    \end{cases}
\end{equation}
 The proposed channel power gain is expressed as a function of the reflector parameters $\kappa_\rho$, $\kappa_z$, $\sigma_z$, and $\ell$, where the interplay among these parameters determines the value of parameter $S$. As $S\to1$, the \gls{NS} becomes smooth, whereas, when $S\to0$, the \gls{NS} becomes increasingly rough and the channel power gain decreases.

 \MD{
\begin{remk}
\label{rem:residual_coupling}
In practical wireless propagation environments, signals may be subject to higher-order interactions, such as inter-reflector mutual coupling and multi-bounce reflections (e.g., secondary reflections between adjacent \glspl{NS}). However, in the \gls{FR3} and \gls{mmWave} regimes, these paths experience compound reflection losses and increased free-space path loss, which results in received power levels that are typically an order of magnitude below those of primary single-bounce reflections. Consequently, while the deterministic terms $c_0\bH_{\NF}^{\LOS}$ and $\bar{c}_r\bar{\bH}_r$ capture the primary, power-dominant propagation mechanisms, any other residual energy, including multi-bounce reflections, diffuse scattering from secondary obstacles, and weak structural edge diffraction \cite{jacob2012diffraction}, can be statistically absorbed into the zero-mean stochastic terms $\tilde{c}_r\tilde{\bH}_r$ without changing the underlying framework.
\end{remk}
}

\subsection{Summary and Discussion}
\label{Summary and Discussion}
The proposed overall \gls{NF} \gls{MIMO} channel model is summarized in \eqref{Eq:MIMO_our_model_location},
\begin{figure*}[t]
\begin{equation}
\label{Eq:MIMO_our_model_location}
	\bH_{\NF} = \underbrace{c_0\bH_{\NF}^\LOS(\bU_0)}_{\text{deterministic}} + \sum_{s=1}^S \underbrace{\hat{c}_s}_{\text{stochastic}}\underbrace{\bH^\scat_s(\hat{\bU}_s)}_{\text{deterministic}}+\sum_{r=1}^R \Big(\underbrace{\bar{c}_{r}\bar{\bH}_{r}(\bU_r)}_{\text{deterministic}}+\underbrace{\tilde{c}_{r}\tilde{\bH}_{r}}_{\text{stochastic}}\Big),\quad
\end{equation}

\begin{subequations}
\label{eq: summarized all channels}
\begin{align}
    &c_0\propto\|\bu_{\tx}-\bu_{\rx}\|^{-1}\!\!,\IEEEyesnumber\IEEEyessubnumber\quad \hat{c}_s\propto(\|\bu_{\rx}-\bu_{s}\|\|\bu_s-\bu_{\tx}\|)^{-1}\!\!,\IEEEyessubnumber \quad \bar{c}_r\propto\|\bu_{\tx}-\bu^r_{\vrx}\|^{-1}\!\!,\IEEEyessubnumber \quad \tilde{c}_r\propto(\|\bu_{\rx}-\bu_c^r\|\|\bu_c^r-\bu_{\tx}\|)^{-1}\!\!,\IEEEyessubnumber\label{eq: summarized all channels a}\\
    &\!\!\!\!\!\!\!\!\!\!\!\!\!\!\!\!\![\bH_{\NF}^\LOS]_{m,n} = \, \e^{\jj\kk\|\bu_{\rx,m}-\bu_{\tx,n}\|}\!\!,\IEEEyessubnumber\quad\! [\bH_s^\scat]_{m,n}=\e^{\jj\kk(\|\bu_{\tx,n}-\bu_{s}\|+\|\bu_{\rx,m}-\bu_{s}\|)}\!\!,\IEEEyessubnumber\quad\! \big[\bar{\bH}_{r}\big]_{m,n} \!\!\!\!\!= \e^{\jj\kk\|\bu^r_{\vrx,m}-\bu_{\tx,n}\|}\!\!,\IEEEyessubnumber\quad\! \big[\tilde{\bH}_{r}\big]_{m,n}\!\!\!\!\!\!\!\!\!\!\sim\sCN(0,1).\IEEEyessubnumber
\end{align}
\end{subequations}
\vspace{0.5em}\hrule
\end{figure*}
where $\bU_0=[\bu_\rx, \bu_\tx]$, $\hat{\bU}_s=[\bu_\tx,\bu_\rx,\bu_s]$, and $\bU_r=[\bu_\tx,\bu_\rx,{\bu^r_\vrx},{\bu^r_\vtx}]$. There are a number of insights that can be extracted from \eqref{Eq:MIMO_our_model_location} and are summarized here:
\begin{itemize}
    \item \textbf{Effective channels:} Let us first compare the channel structures and the corresponding coefficients in \eqref{eq: summarized all channels}.
By considering the expressions in \eqref{eq: summarized all channels a}, we typically have
\[
\underbrace{c_0>\bar{c}_r}_{\text{Effective channels}}\overset{(a)}{\gg}\underbrace{\tilde{c}_r\overset{(b)}{>}\hat{c}_s}_{\text{Scattering channels}},
\]
where inequality $(a)$ holds due to the double path loss for scattering channels \gls{wrt} the single path loss of the effective channels. In addition, inequality $(b)$ holds due to the larger area of the reflector compared to the point scatter.
\item \textbf{Features of the structures:} Interestingly, the generalized \gls{NF} \gls{MIMO} channel model in (\ref{Eq:MIMO_our_model_location}) reveals that all effective channel components, i.e., the elements of $\bH^\LOS_\NF$ and $\bar{\bH}_r$, are functions of the real or virtual positions of the \gls{Tx} and \gls{Rx} antennas. Furthermore, the key difference between a reflector scatterer and a point scatterer lies in the nature of the corresponding channel structures, as the former leads to a deterministic channel component, while the latter results in a purely stochastic one. By defining angular spread $\theta_c$ in Corollary~\ref{corol: sinc} in terms of the distance and area of the reflector, we can characterize the spatial correlation. This, in turn, allows us to determine whether a given scatterer behaves as a point scatterer or a reflector scatterer in a specific setup.
\item \textbf{\gls{NF} beam focusing:} The aforedescribed \gls{Tx}/\gls{Rx} real and virtual positions play a crucial role in enabling efficient \gls{NF} beam training \cite{delbari2026far,liu2023near,lu2023near,ramezani2023near,NF_beam_tracking}. In practice, \gls{NF} beam training allows the \gls{Tx} and \gls{Rx} to learn and focus beams toward the key locations \(\bu_\tx\), \(\bu_\rx\), \(\bu^r_\vrx\), and \(\bu^r_\vtx\) \cite{George2022}. Once the \gls{NF} beamformer leads to a sufficiently high channel power gain, conventional pilot-based methods can be used to estimate the effective end-to-end channel \cite{jamali2022}.
\item \textbf{Generalized \gls{NF} \gls{MIMO} Rician model:} Let us first rewrite \eqref{Eq:MIMO_our_model_location} in a form similar to the Rician channel model:
\begin{equation}
\label{Eq:MIMO_our_model_location Similar Ricean}
	\!\!\!\!\bH_{\NF}\! =\! c_0\!\!\left(\!\!\bH_{\NF}^\LOS \!\!+\!\! \sum_{s=1}^S \!\hat{k}_s\bH^\scat_s\!\!+\!\!\sum_{r=1}^R \!\left(\bar{k}_{r}\bar{\bH}_{r}\!+\!\tilde{k}_{r}\tilde{\bH}_{r}\right)\!\!\right)\!,
\end{equation}
where $\hat{k}_s\propto \frac{\hat{l}_s\|\bu_{\tx}-\bu_{\rx}\|}{\|\bu_{\rx}-\bu_{s}\|\|\bu_s-\bu_{\tx}\|}$, $\bar{k}_r\propto \frac{\bar{l}_r\|\bu_{\tx}-\bu_{\rx}\|}{\|\bu_{\vrx}-\bu_{\tx}\|}$, and $\tilde{k}_r\propto \frac{\tilde{l}_r\|\bu_{\tx}-\bu_{\rx}\|}{\|\bu_{\rx}-\bu_c^r\|\|\bu_c^r-\bu_{\tx}\|}$ with $\hat{l}_s$, $\bar{l}_r$, and $\tilde{l}_r$ incorporating any losses of the channels except the distance. The generalized Rician model in \eqref{Eq:MIMO_our_model_location Similar Ricean} for \gls{NF} scenarios comprises the \gls{LOS} channel as well as the different \gls{nLOS} components generated by point scatterers and surface reflectors, respectively. The values of the generalized Rician factors, $\hat{k}_s$, $\bar{k}_r$, and $\tilde{k}_r$, depend on the specific scenario and can be chosen based on simulation or experimental data.
\MD{\item \textbf{Point scattering vs. reflectors:} 
While the \gls{HF} integral in \eqref{eq: Huygen} accurately describes reflection from an object, the point scattering model and the surface reflection model become valid approximations, respectively, for small and large aperture sizes $|\Uset|$. For ideal surface reflection (i.e., large $|\Uset|$ and $g=0$), the channel gain can be obtained from Theorem~\ref{Theorem Gaussian distribution} as $\bar{c}_r(0)$ in \eqref{eq:coherent component a}. For small $|\Uset|$, the \gls{HF} integral $c_\mathrm{HF}(\Uset)$ in \eqref{eq: Huygen} can be approximated as follows
\begin{align}
    &\hat{c}_s(\Uset)=\frac{\zeta}{\jj\lambda}\times|\Uset|\times\\
    & \frac{z_{\tx,n}}{\|\bu_s-\bu_{\tx,n}\|^2}\frac{z_{\rx,m}}{\|\bu_{\rx,m}-\bu_s\|^2}
\e^{\jj \kappa\|\bu_s-\bu_{\tx,n}\|}\e^{\jj \kappa\|\bu_{\rx,m}-\bu_s\|}\nonumber.
\end{align}
To compare the accuracy of the point scattering and surface reflection models, we plot in Fig.~\ref{fig:c_perfect} $\frac{|\hat{c}_s(\Uset)|}{|c_\mathrm{HF}(\Uset)|}$ and $\frac{|\bar{c}_r(0)|}{|c_\mathrm{HF}(\Uset)|}$ as functions of $|\Uset|$ for $f=28$~GHz and $60$~GHz. As can be observed from this figure, the surface reflection model becomes accurate for $\sqrt{|\Uset|}\geq 30$~cm and $\sqrt{|\Uset|}\geq 50$~cm for $f=60$~GHz and $28$~GHz, respectively, which is the practical regime for dominant reflectors such as ceiling, floor, and walls.

\begin{figure}
    \centering
    \includegraphics[width=0.45\textwidth]{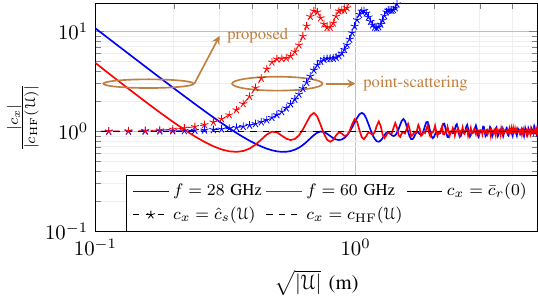}
    \caption{\MD{Normalized received signal amplitude of the proposed model and point scattering model vs. the length of a square reflector $\sqrt{|\mathcal{U}|}$ for $f = 28$~GHz and $f = 60$~GHz for an ideal smooth surface ($g = 0$).}}
    \label{fig:c_perfect}
\end{figure}

In addition to the accuracy of the channel, we analyze the impact the different channel models have on beamformer design in Fig. \ref{fig:sizeA_Tx}. This figure shows the normalized received signal amplitude when the beamformer is designed based on the proposed reflection model and the point scattering models as a function of $\sqrt{|\mathcal{U}|}$ for $f = 28$~GHz and $N_\tx=256$ and $N_\rx=1$ for a smooth square surface ($g=0$). As can be seen, while the point scattering model yields better performance for $\sqrt{|\Uset|}\leq 30$~cm, a significant performance loss occurs when the beamformer is designed wrongly based on the point scattering model for practically large reflector sizes $\sqrt{|\Uset|}\geq 0.5$~m.

\begin{figure}
    \centering
    \includegraphics[width=0.45\textwidth]{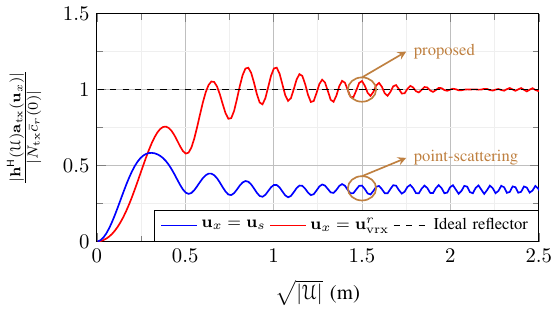}
    \caption{\MD{Normalized received signal amplitude when the beamformer is designed based on the proposed model and the point scattering model vs. the length of a square reflector $\sqrt{|\mathcal{U}|}$ for $f = 28$~GHz and $N_\tx=
    256$ and $N_\rx=1$ for an ideal smooth surface ($g = 0$).}}
    \label{fig:sizeA_Tx}
\end{figure}
}
\MD{\item \textbf{Complexity and performance trade-off:} While generating channel realizations using the proposed \gls{NF} \gls{MIMO} channel model incurs slightly higher computational complexity than conventional point-scattering models, it yields substantial performance and modeling gains, as demonstrated in Figs.~ \ref{fig:c_perfect}, \ref{fig:sizeA_Tx}, and in Section~\ref{sec: Performance Evaluation}. Conversely, although numerical physics-based surface integration models (e.g., \gls{HF}) \cite{Hoppe2017} and site-specific \gls{RT} simulations \cite{Yuan2024,Choi2023} can achieve high accuracy for a fixed scenario, they are computationally intensive, requiring dense sub-wavelength surface discretization or exhaustive ray-path searches, and lack analytical tractability. In contrast, the proposed channel model is fully parameterized by a small set of physical parameters, which provides closed-form statistical insights. This allows low-complexity channel synthesis across diverse spatial topologies while maintaining high electromagnetic fidelity relative to \gls{HF} benchmarks, which makes it well suited for algorithmic beam design and system-level optimization.}
\end{itemize}

\section{When Is Exploiting NLOS Paths Beneficial?}
\label{sec: When Exploiting nLOS Paths Is Beneficial?}
\begin{figure}
    \centering
    \includegraphics[width=0.45\textwidth]{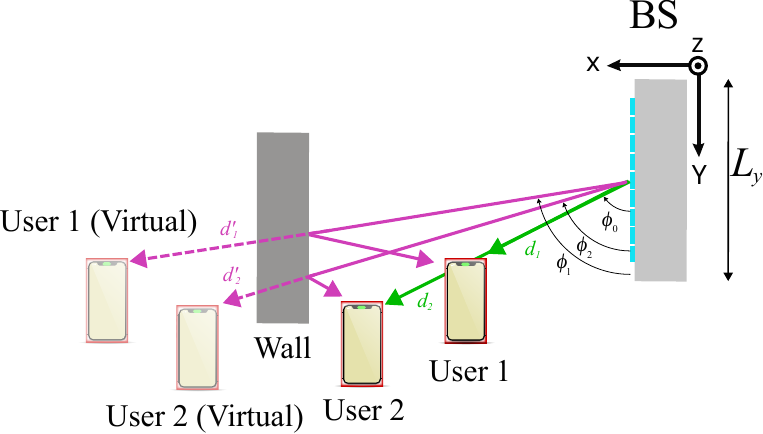}
    \caption{Two-dimensional representation of the three-dimensional scenario.}
    \label{fig:system 2D}
    \vspace{-0.5cm}
\end{figure}
In Section~\ref{Summary and Discussion}, we introduced the proposed \gls{NF} \gls{MIMO} channel model and argued that \gls{NF} beam focusing can effectively identify dominant propagation paths to serve \glspl{MU}. In this section, we compare the performance achieved by exploiting \gls{LOS} and dominant \gls{nLOS} paths in multi-user scenarios based on the corresponding \glspl{SINR}. Although the channel gain of an \gls{nLOS} path is generally weaker (more attenuated) than that of a \gls{LOS} path, it can still be exploited to spatially separate \glspl{MU} using \gls{SDMA}. The effectiveness of this approach, however, depends on the spatial resolution of the paths, which is primarily influenced by: (i) the angular separation of the paths, and (ii) the beamwidth achievable by the \gls{BS}, which is in turn constrained by its array size. This leads to a fundamental trade-off: Using \gls{nLOS} paths offers improved spatial separability at the cost of reduced channel gain, while relying on the stronger \gls{LOS} path may result in limited spatial resolution. In the remainder of this section, we analyze this trade-off through a detailed example.

To obtain insight, we study a simple but representative setting where two \glspl{MU} are served by a \gls{BS} at the same time and frequency\footnote{In a practical scenario, there may be multiple \glspl{MU}. In this case, \glspl{MU} may be first grouped such that those within the same group share the same resource. Different groups can then be served using orthogonal resources, enabled by \gls{TDMA} or \gls{FDMA}. In this work, we only consider a single group of two \glspl{MU}.}, and both are located at the same angular dimension from the \gls{BS}, as illustrated in Fig.~\ref{fig:system 2D}. It is known that, while for this scenario, spatial multiplexing cannot be achieved in the \gls{FF} regime, it may still be possible to spatially separate \glspl{MU} in the \gls{NF} regime, nonetheless, the corresponding \gls{SINR} significantly depends on the \gls{BS}-\gls{MU} and \gls{MU}-\gls{MU} distances \cite{ramezani2023near}. We are interested in studying whether and when the additional spatial resolution offered by the \gls{nLOS} links is beneficial in this challenging \gls{NF} scenario.

We consider a two-dimensional setting for simplicity to be able to provide insightful analytical results\footnote{In Section~\ref{sec: Performance Evaluation}, we will consider a realistic 3D system and numerically demonstrate the regime where the use of \gls{nLOS} \gls{NF} paths is beneficial.}.
Assume the \gls{BS} is equipped with two \glspl{ULA}, each of length $L_y=N_\tx\frac{\lambda}{2}$ and aligned along the $\y$-axis, where the upper array is steered to serve the first \gls{MU}, while the lower array serves the second \gls{MU}. We theoretically analyze the \gls{SINR} of the first \gls{MU} where the data from the second \gls{MU} is treated as interference (a similar analysis can be done for the second \gls{MU}, which is omitted here due to space constraints). We further assume that the \gls{BS} transmit power is divided equally between both \glspl{MU} ($P=P_t/2$), and $N_\rx=1$ for both \glspl{MU} in \eqref{Eq:IRSbasic}. This leads to
\begin{equation}
\label{eq: SINR}
    \SINR_1=\frac{|\bh_1^\Herm\bq_1|^2P}{|\bh_1^\Herm\bq_{2}|^2P+\sigma_n^2}=\frac{|\bh_1^\Herm\bq_1|^2}{|\bh_1^\Herm\bq_{2}|^2+\frac{\sigma_n^2}{P}}.
\end{equation}
To simplify the analysis and gain intuition, we consider \gls{NF} beamforming, where
$\bq_k=\sqrt{\frac{1}{N_\tx}}[\e^{\jj \omega_{1}^{(k)}}, \ldots, \e^{\jj \omega_{N_\tx}^{(k)}}]^\Trans,\,\forall k\in\{1,2\}$. Assuming $R=1$ and $S=0$ in \eqref{Eq:MIMO_our_model_location Similar Ricean} and approximating $\tilde{k}_r\approx0$, the channel for the first \gls{MU} becomes $\bh_1=c_0(\bh_1^\LOS+\bar{k}_1\bar{\bh}_1)$, where $(\bh_1^\LOS)^\Herm=\sqrt{\frac{1}{N_\tx}}\ba_\tx^\Trans(\bu_\rx^{(1)})$, with $\bu_\rx^{(1)}$ denoting the position of the first \gls{MU}, and $\bar{\bh}_1^\Herm=\sqrt{\frac{1}{N_\tx}}\ba_\tx^\Trans(\bu_\vrx^{(1)})$, with $\bu_\vrx^{(1)}$ representing the virtual (reflected) location of the first \gls{MU} on the wall. 

In the following, we analyze \(\SINR_1\) under two different scenarios:
\textit{i) \gls{LOS}:} The \gls{BS} exploits only the \gls{LOS} paths for beamforming.
\textit{ii) \gls{nLOS}:} The \gls{BS} transmits exclusively along the \gls{nLOS} paths\footnote{We exclude the mixed case in which one \gls{MU} is served via an \gls{LOS} path and the other via a \gls{nLOS} path, since the resulting \gls{SINR} for each \gls{MU} is comparable to one of the two considered scenarios.}. 

\subsection{Scenario \textit{i)}: Exploiting Only LOS Paths}

In scenario \textit{i)},  each array serves the corresponding \gls{MU} via the \gls{LOS} path. Referring to Fig.~\ref{fig:system 2D}, let the separation between the \glspl{MU} be
$d \triangleq |d_2 - d_1|$, with $d_1$ and $d_2$ representing the distances from the center of the \gls{BS} to \gls{MU} 1 and \gls{MU} 2, respectively. We optimize \(\bq_1\) and \(\bq_2\) to focus exclusively on the \gls{LOS} paths of \gls{MU}~1 and \gls{MU}~2, respectively, and denote the corresponding beamformers by \(\bq_1^{\LOS}\) and \(\bq_2^{\LOS}\). Substituting these into \eqref{eq: SINR}, yields:
\begin{equation}
    \label{eq: SINR 2}
\SINR_1^\LOS=\frac{|h_{11}^\LOS|^2+|\hat{h}_{11}^\nLOS|^2}{|h_{12}^\LOS|^2+|\hat{h}_{12}^\nLOS|^2+\frac{\sigma^2_n}{P}},
\end{equation}
where we defined $|h_{11}^\LOS|^2\defeq|c_0(\bh_1^\LOS)^\Herm\bq_1^\LOS|^2$, $|h_{12}^\LOS|^2\defeq|c_0(\bh_1^\LOS)^\Herm\bq_{2}^\LOS|^2$, and the superscript $(\cdot)^\LOS$ indicates the usage of the \gls{LOS} link in the \gls{BS}-\gls{MU} channel. The notation $\hat{(\cdot)}$ refers to contributions from side lobes (through \gls{nLOS} path in scenario \textit{i)}), i.e., $|\hat{h}_{11}^\nLOS|^2=|c_0\bar{k}(\bar{\bh}_1)^\Herm\bq_1^\LOS|^2$ and $|\hat{h}_{12}^\nLOS|^2=|c_0\bar{k}(\bar{\bh}_1)^\Herm\bq_2^\LOS|^2$. We neglect the power received through the side lobes, assuming the \gls{BS} array is sufficiently large\footnote{We note that side lobes of both the desired signal and the interference may also pass through the \gls{nLOS} link. However, if $L_y$ is sufficiently large, their impact becomes negligible \gls{wrt} the main lobes, which are the primary focus of our analysis here for both scenarios \textit{i)} and \textit{ii)}. This observation is also confirmed by the simulation results shown in Fig.~\ref{fig:N_prop}.}, i.e., $|\hat{h}_{11}^\nLOS|^2\approx|\hat{h}_{12}^\nLOS|^2\approx0$.

Assuming each array optimizes its element phases to coherently combine at the location of the intended \gls{MU} in the \gls{NF}, we set $\omega_n^{(k)}=-\kk\|\bu_{\tx,n}-\bu_\rx^{(k)}\|$ and expand the phase as \cite[Lemma~1]{delbari2026far}
\ifarxiv
\begin{equation}
\label{eq: phase shift}
    -\kappa d_k\!\!\left(\!\!\underbrace{\frac{\|\bu_{\tx,n}\|}{d_k}\cos\phi_0}_{\text{Linear term}}\!+\!\underbrace{\left(\frac{\|\bu_{\tx,n}\|}{d_k}\right)^2\!\!\frac{\sin^2\phi_0}{2}}_{\text{Quadratic term}}\!+\!\!\underbrace{\bigo\!\!\left(\!\frac{\|\bu_{\tx,n}\|}{d_k}\!\right)^3\!}_{\text{Non-quadratic terms}}\right)\!,
\end{equation}
where $k\in\{1,2\}$. 
\fi
Substituting \eqref{eq: phase shift} into $\bq_1^\LOS$ and $\bq_2^\LOS$ for each \gls{MU}, we obtain
\begin{equation}
    \label{eq: SINR 3}
\SINR_1^\LOS=\frac{|c_0|^2\left|\frac{1}{N_\tx}
\sum_{n=1}^{N_\tx}\e^{\jj(\kk\|\bu_{\tx,n}-\bu_\rx^{(1)}\|+\omega_n^{(1)})}\right|^2}{|c_0|^2\left|\frac{1}{N_\tx}\sum_{n=1}^{N_\tx}\e^{\jj(\kk\|\bu_{\tx,n}-\bu_\rx^{(1)}\|+\omega_n^{(2)})}\right|^2+\frac{\sigma^2_n}{P}}.
\end{equation}
We assume the cubic and higher‑order terms (non-quadratic terms) are negligible for both \glspl{MU}\footnote{See \cite[Fig.~3]{delbari2026far} for a quantitative range (frequency, aperture size, and distances) where this approximation holds.}. Under this assumption, the numerator of \eqref{eq: SINR 3} is approximately $|c_0|^2$. In the denominator, the linear terms for the two \glspl{MU} cancel, leaving only the quadratic mismatch. Approximating the discrete sum by an integral (large‑\(N_\tx\) continuous‑aperture model), we obtain
\begin{align}
    \label{eq: SINR 4}
\SINR_1^\LOS&=\left(\Big|\frac{1}{L_y}\int_{-\frac{L_y}{2}}^{\frac{L_y}{2}}\e^{\jj a_1 y^2} \dd y\Big|^2+\frac{\sigma^2_n}{P|c_0|^2}\right)^{-1}\nonumber\\
&=\left(\frac{\pi}{a_1L_y^2}\Big|\erfi\Big(\sqrt{\jj a_1}\frac{L_y}{2}\Big)\Big|^{2}+\frac{\sigma^2_n}{P|c_0|^2}\right)^{-1}\!\!,
\end{align}
where $\erfi(\cdot)$ is the imaginary error function and $a_1=\kappa \frac{\sin^2\phi_0}{2}\frac{d}{d_1d_2}$. As observed from \eqref{eq: SINR 4}, increasing either $a_1$ or $L_y$ reduces the interference term. However, achieving this requires increasing the distance between the \glspl{MU} ($d$) or decreasing their respective distances from the \gls{BS} ($d_1$ and $d_2$). In the following section, we show that these requirements can be relaxed by exploiting the \gls{nLOS} path.

\subsection{Scenario \textit{ii)}: Exploiting Only NLOS Paths}
In scenario \textit{ii)}, each array communicates with its designated \gls{MU} solely through the \gls{nLOS} path. As illustrated in Fig.~\ref{fig:system 2D}, the virtual positions of the two \glspl{MU} are located at different azimuth angles $\phi_1$ and $\phi_2$ \gls{wrt} the \gls{BS}. We design \(\bq_1\) and \(\bq_2\) to exclusively target the \gls{nLOS} paths of \gls{MU}~1 and \gls{MU}~2, respectively, and denote the corresponding beamformers as \(\bq_1^{\nLOS}\) and \(\bq_2^{\nLOS}\). Substituting these beamformers into \eqref{eq: SINR} gives
\begin{equation}
    \label{eq: SINR 2 2}
\SINR_1^\nLOS=\frac{|h_{11}^\nLOS|^2+|\hat{h}_{11}^\LOS|^2}{|h_{12}^\nLOS|^2+|\hat{h}_{12}^\LOS|^2+\frac{\sigma^2_n}{P}},
\end{equation}
where $|h_{11}^\nLOS|^2 \triangleq |c_0\bar{k}_1\bar{\bh}_1^\Herm\bq_1^\nLOS|^2, \quad
|h_{12}^\nLOS|^2 \triangleq |c_0\bar{k}_1\bar{\bh}_1^\Herm\bq_{2}^\nLOS|^2,$
and the superscript \((\cdot)^\nLOS\) indicates that the \gls{nLOS} link of the \gls{BS}–\gls{MU} channel is used. The notation $\hat{(\cdot)}$ represents contributions arising from side lobes via the \gls{LOS} path in this scenario. We again disregard the power received via side lobes, under the assumption that the \gls{BS} array is sufficiently large, i.e., $|\hat{h}_{11}^\LOS|^2=|c_0(\bh_1^\LOS)^\Herm\bq_1^\nLOS|^2\approx0$ and $|\hat{h}_{12}^\LOS|^2=|c_0(\bh_1^\LOS)^\Herm\bq_2^\nLOS|^2\approx0$. The validity of this approximation is confirmed by our simulation results shown in Fig.~\ref{fig:N_prop}.

We again assume each array optimizes its element phases to be coherently added at the corresponding \gls{MU} in the \gls{NF} regime. We write $\omega_n^{(k)}=-\kk\|\bu_{\tx,n}-\bu_\rx^{(k)}\|$ and expand the phase as
\ifarxiv
\begin{equation}
\label{eq: phase shift 2}
    -\kappa d_k\!\!\left(\!\!\underbrace{\frac{\|\bu_{\tx,n}\|}{d_k}\cos\phi_k}_{\text{Linear term}}\!+\!\underbrace{\left(\frac{\|\bu_{\tx,n}\|}{d_k}\right)^2\!\!\frac{\sin^2\phi_k}{2}}_{\text{Quadratic term}}\!+\!\!\underbrace{\bigo\!\!\left(\!\frac{\|\bu_{\tx,n}\|}{d_k}\!\right)^3\!}_{\text{Non-quadratic terms}}\right)\!\!,
\end{equation}
where $k\in\{1,2\}$. Substituting \eqref{eq: phase shift 2} into $\bq_1^\nLOS$ and $\bq_2^\nLOS$ for each \gls{MU}, we obtain
\fi
\begin{equation}
    \label{eq: SINR 3 2}
\SINR_1^\nLOS=\frac{|c_0\bar{k}_1|^2\left|\frac{1}{N_\tx}
\sum_{n=1}^{N_\tx}\e^{\jj(\kk\|\bu_{\tx,n}-\bu_\rx^{(1)}\|+\omega_n^{(1)})}\right|^2}{|c_0\bar{k}_1|^2\left|\frac{1}{N_\tx}\sum_{n=1}^{N_\tx}\e^{\jj(\kk\|\bu_{\tx,n}-\bu_\rx^{(1)}\|+\omega_n^{(2)})}\right|^2+\frac{\sigma^2_n}{P}}.
\end{equation}
By neglecting the cubic and higher‑order terms for both \glspl{MU}, the numerator of \eqref{eq: SINR 3 2} is approximately $|c_0\bar{k}_1|^2$. In the denominator, this time the linear terms for the two \glspl{MU} do not cancel each other since $\phi_1\neq\phi_2$. Similar to scenario \textit{i)}, by approximating the discrete sum by an integral (large‑\(N_\tx\) continuous‑aperture model), we obtain

\begin{align}
    \label{eq: SINR 4 2}
\SINR_1^\nLOS&=\left(\Big|\frac{1}{L_y}\int_{-\frac{L_y}{2}}^{\frac{L_y}{2}}\e^{\jj a_2y^2+\jj b y} \dd y\Big|^2+\frac{\sigma^2_n}{P|c_0\bar{k}_1|^2}\right)^{-1}\nonumber\\
&\!\!\!\!\!\!\!\!\!\!\!\!\!\!\!\!\!\!\!\!\!\!\!\!\!\!\!\!=\left(\frac{\pi}{4a_2L_y^2}\Big|\erfi\Big(\sqrt{\frac{\jj}{4a_2}}(2a_2y+b)\Big)\big|_\frac{-L_y}{2}^\frac{L_y}{2}\Big|^{2}+\frac{\sigma^2_n}{P|c_0\bar{k}_1|^2}\right)^{-1}\nonumber\\
&\overset{(a)}{=}\left(\sinc^{2}\left(\frac{L_y}{2\pi}b\right)+\frac{\sigma^2_n}{P|c_0\bar{k}_1|^2}\right)^{-1},
\end{align}
where $a_2=\frac{\kappa}{2}\Big|\frac{\sin^2(\phi_1)}{d_1'}-\frac{\sin^2(\phi_2)}{d_2'}\Big|$ and $b=\kappa(\cos(\phi_1)-\cos(\phi_2))$.
Moreover, $d_1'$ and $d_2'$ denote the \gls{MU} distances relevant for the \gls{nLOS} contributions. Approximation $(a)$ holds under the condition $a_2\ll b$. As observed from \eqref{eq: SINR 4 2}, increasing either $b$ or $L_y$ reduces the interference term. Exploiting the \gls{nLOS} path increases the (virtual) spatial separation between the \glspl{MU}, which in turn leads to a larger value of $b$. 

For practical parameter settings, the interference term in~\eqref{eq: SINR 4 2} decays faster than that in~\eqref{eq: SINR 4}. We compare the \glspl{SINR} of the first \gls{MU} in Figs.~\ref{fig: k-distance_d} and~\ref{fig: k-distance_d1}, which illustrate the trade-off between exploiting the \gls{nLOS} path and relying solely on the \gls{LOS} path. The former may be weaker in power but provides higher spatial resolution. As can be seen from this figure, exploiting \gls{nLOS} paths leads to a higher \gls{SINR} in most cases for the adopted parameters. This specific observation, while limited to this configuration, motivates a broader investigation into whether \gls{nLOS} paths can also be beneficial in more general settings, e.g., 3D scenario, discussed in Section~\ref{sec: Performance Evaluation}.
\begin{figure}[t]
\centering
\begin{subfigure}{0.45\textwidth}
\includegraphics[width=\textwidth,height=0.55\textwidth]{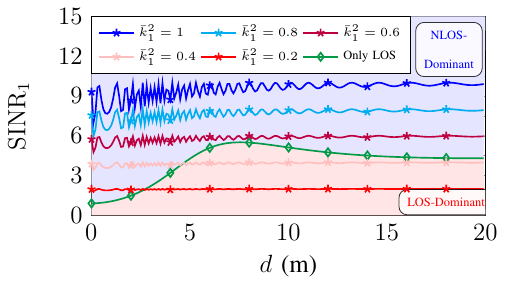}
\caption{$\SINR_1$ vs. the distance between the two \glspl{MU} when $d_1=5\sqrt{2}$.}
\label{fig: k-distance_d}
\end{subfigure}
\begin{subfigure}{0.45\textwidth}
\includegraphics[width=\textwidth,height=0.55\textwidth]{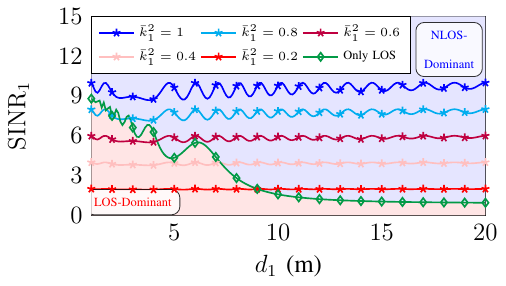}
    \caption{$\SINR_1$ vs. the distance of BS to the first \gls{MU} and $d=5$.}
    \label{fig: k-distance_d1}
\end{subfigure}
\caption{SINR of the first \gls{MU} for different \glspl{MU} locations where the distance of the \gls{NS} and second \gls{MU} is fixed to 1~m, $\frac{\sigma_n^2}{P|c_0|^2}=0.1$, and they are located at the same angular position from the BS.}
\vspace{-0.5cm}
\label{fig: k-distance}
\end{figure}
\MD{
\begin{remk}
    Note that Figs.~\ref{fig: k-distance_d} and \ref{fig: k-distance_d1} are plotted under the assumption that the non-quadratic terms, i.e., $\mathcal{O}\left(\|\bu_{\tx,n}\|/d_k\right)^3$, in both \eqref{eq: SINR 3} and \eqref{eq: SINR 3 2} are negligible. However, this assumption breaks down at extremely high frequencies (e.g., in the THz band) or at very short propagation distances, where these non-quadratic terms introduce significant phase deviations. To quantitatively illustrate the impact of neglecting these higher-order terms, Fig.~\ref{fig:NF_THZ} plots the normalized error between the analytical \gls{SINR} obtained in \eqref{eq: SINR 4} and the exact numerical \gls{SINR}, denoted by $\widetilde{\SINR}^\LOS_1$, where the higher-order error terms remain in the numerator and denominator of \eqref{eq: SINR 3}. As observed, the error for the \gls{LOS} analysis remains negligible across all considered scenarios for frequencies up to $f = 30$~GHz. Beyond this frequency, the error increases when the link distance is extremely small (e.g., $d_1 = 5\sqrt{2}$~m) and the \gls{BS} aperture size is large (e.g., $L_y = 1$~m, corresponding to $N_\tx = 400$ at $f = 60$~GHz). In these scenarios, higher-order terms (i.e., the cubic term) or the full spherical wave must be considered.

    \begin{figure}[htbp]
    \centering
    \includegraphics[width=0.45\textwidth]{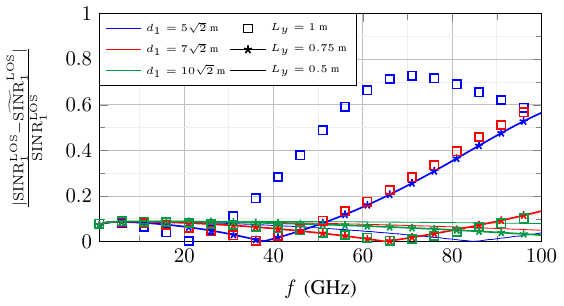}
    \caption{\MD{Normalized error of the analytical SINR in Section~IV vs. frequency for different link distances $d_1$ and BS aperture lengths $L_y$.}}
    \label{fig:NF_THZ}
\end{figure}
\end{remk}
}


\section{Performance Evaluation}
\label{sec: Performance Evaluation}
First, we validate the proposed channel models through simulations based on the numerical evaluation of the \gls{HF} integral. Then, we investigate a practical scenario where a \gls{BS} serves two \glspl{MU} through both \gls{LOS} and \gls{nLOS} paths.

\begin{remk}
The source code is available online at \href{https://github.com/MohamadrezaDelbari/NF-Multipath-MIMO-Channels}{\textcolor{blue}{https://github.com/MohamadrezaDelbari/NF-Multipath-MIMO-Channels}}.
\end{remk}

\subsection{Model Verification}
\label{Simulation Result}
\textbf{Simulation Setup:}
We consider a $3\times3$~m$^2$ \gls{NS} lying in the $\x-\y$ plane, where the surface height variation $Z$ follows a Gaussian distribution, i.e., $Z\sim\Nset(0,\sigma_z^2)$. The carrier frequency is set to $28$~GHz, i.e., $\lambda\approx1.1~$cm. At this frequency, the far-field region of the \gls{NS} begins at $\frac{2D^2}{\lambda} = 3288$~m, where $D$ denotes the largest dimension of the \gls{NS}. Moreover, a single \gls{Tx} is located at coordinates $(0, 0, 90)$, and we consider a single \gls{Rx} whose location varies and is specified individually in each corresponding figure. For all considered scenarios, we assume a passivity factor of $\zeta = 1$ for the \gls{NS}. In the following, we investigate several aspects to validate the accuracy and applicability of the proposed \gls{NF} \gls{MIMO} channel model. 

\textbf{Verification of $\Ex\{c_r\bH_r^\refl\}$ and $\Ex\{|c_r|\}$ in Theorem~\ref{Theorem Gaussian distribution} and (\ref{eq: stochastic power}):} Fig.~\ref{fig:3regimes} illustrates a single realization of the real and imaginary components, \(\Re/\Im\{c_r[\bH_r^\refl]_{m,n}\}\), as well as magnitude $|c_r|$ and the average (Avg.) over 100 normalized samples. These results are obtained by numerically evaluating the \gls{HF} integral in \eqref{eq: Huygen}, and are shown using solid lines. In addition, the theoretical predictions based on Theorem~\ref{Theorem Gaussian distribution} and \eqref{eq: total power} are indicated using square markers. Both the \gls{HF} integral-based and the theoretical results are normalized by the deterministic component when $g=0$, i.e., $\bar{c}_{r}(0)$. The curve corresponding to $\Ex\{|c_r|\}$ is computed via \eqref{eq: total power}. As shown in Fig.~\ref{fig:3regimes}, the averaged results obtained from the \gls{HF} integral (solid lines) are in perfect agreement with the proposed theoretical predictions (square markers). Specifically, the averages of the real and imaginary components of $\frac{c_r[\bH_r^\refl]_{m,n}}{\bar{c}_{r}(0)}$, computed via the \gls{HF} integral in \eqref{eq: Huygen}, follow a decay slope of $\e^{-\frac{g}{2}}$, consistent with the theoretical expression $\frac{\bar{c}_{r}(g)}{\bar{c}_{r}(0)}$ in \eqref{eq:coherent component a}. Furthermore, the average value of $\frac{|c_r|}{\bar{c}_{r}(0)}$ obtained from the \gls{HF} integral converges asymptotically to $\frac{|c_{n,r,+\infty}|}{|\bar{c}_{r}(0)|}$ as $\kappa\sigma_z$ becomes sufficiently large.

\textbf{Verification of the PDF of the Elements of $\tilde{\bH}_{r}$:} Fig.~\ref{fig:distribution} presents the \glspl{PDF} of $\Re/\Im\{c_r[\bH_r^\refl]_{m,n}\}$ for $\kappa\sigma_z = 0,\; 0.5,\; 3$, obtained via numerical evaluation of the \gls{HF} integral in \eqref{eq: Huygen}. As shown, the histograms, computed over multiple realizations of the surface profile $z(x,y)$, closely match the Gaussian distribution for all three regimes. This excellent agreement supports the validity of modeling the reflection coefficients as Gaussian-distributed random variables.

\textbf{Verification of the Correlation of the Elements of $\tilde{\bH}_{r}$ in Corollary \ref{corol: sinc}:} Fig.~\ref{fig:corr} illustrates the spatial correlation between two \gls{Rx} antennas as a function of their distance. As observed, the results derived in Corollary~\ref{corol: sinc} closely approximate those obtained with the \gls{HF} integral. The spatial correlation is influenced by the surface roughness, $\sigma_z$, the angle of view $\theta_c$, and the relative locations of the antennas. We study two specific configurations addressed in Corollary~\ref{corol: sinc}: One where the surface normal vector $\vec{\bn}$ is perpendicular to the antenna displacement vector $\bu_R - \bu_{R_p}$, i.e., $\vec{\bn}\perp\bu_R-\bu_{R_p}$, and another one where $\vec{\bn}$ is parallel to $\bu_R - \bu_{R_a}$, i.e., $\vec{\bn}\parallel\bu_R-\bu_{R_a}$; see Fig.~\ref{fig:mixed2} for illustration. In the first configuration, the correlation decays more rapidly with distance, which is consistent with the theoretical predictions in Corollary~\ref{corol: sinc}. Moreover, consistent with Corollary~\ref{corol: sinc}, we observe that the spatial correlation between the antennas increases as the angle of view, $\theta_c$, decreases.

\textbf{Verification of the impact of length correlation in Theorem~\ref{theorem power regime 3} and \eqref{eq: S function}:} Fig.~\ref{fig:S analysis} presents a single realization and the average channel power gain, both obtained by numerically evaluating the \gls{HF} integral in \eqref{eq: Huygen}. These results are depicted with solid lines. Additionally, theoretical predictions based on Theorem~\ref{theorem power regime 3} and \eqref{eq: S function} are shown using square markers. Both results are normalized by $|\bar{c}_{r}(0)|^2$. The theoretical results are in good agreement with the numerical results, seen in Fig.~\ref{fig:S analysis}. 
\begin{figure}
    \centering
    \centering
\includegraphics[width=0.45\textwidth]{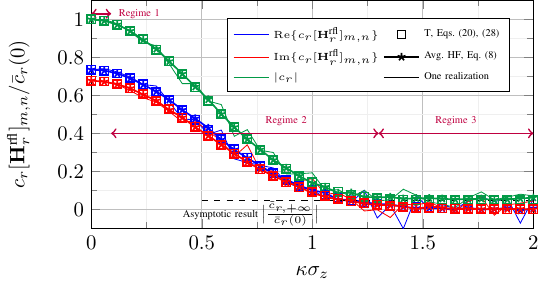}
    \caption{One realization, average, and theoretical results for the real part, the imaginary part, and the absolute value of the channel amplitude $c_r[\bH_r^\refl]_{m,n}$ normalized to ${\bar{c}_{r}(0)}$. Based on the value of $\kappa\sigma_z$, the three different regimes, namely Regime 1 (SR), Regime 2 (transient regime), and Regime 3 (SS), are distinguished. Here, T indicates the \underline{t}heoretical result.}
    \label{fig:3regimes}
\end{figure}

    \begin{figure}[t]
\centering
    \includegraphics[width=0.45\textwidth]{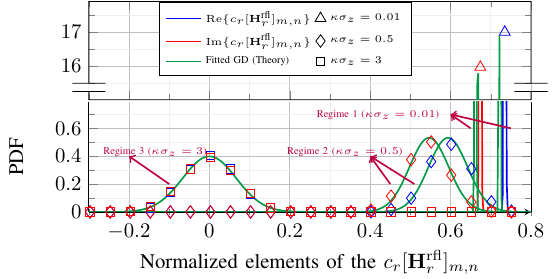}
    \caption{The distribution of the real and imaginary parts of $c_r[\bH_r^\refl]_{m,n}$ obtained with the HF integral in (\ref{eq: Huygen}) normalized by $\bar{c}_{r}(0)$. Corresponding fitted \glspl{GD} are also shown.}
    \label{fig:distribution}
    \vspace{-0.4 cm}
    \end{figure}
    
    \begin{figure}[t]
        \centering
\includegraphics[width=0.45\textwidth]{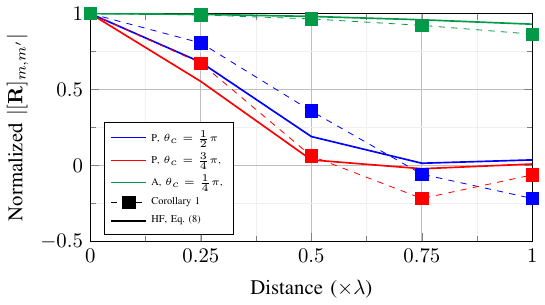}
        \caption{Comparison of the normalized spatial correlation for different positions of the Rx antennas. For conciseness, we use the abbreviation (P, A), indicating \underline{P}erpendicular and \underline{A}ligned Rx antennas w.r.t. $\vec{\bn}$, respectively. For all results, $\kappa\sigma_z=3$ was used.}
        \label{fig:corr}
        \vspace{-0.4 cm}
    \end{figure}

    \begin{figure}[t]
        \centering
        \includegraphics[width=0.45\textwidth]{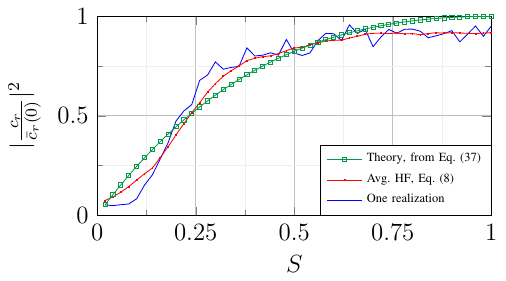}
        \caption{One realization, average, and theoretical results for the channel power gain as a function of $S$.}
        \label{fig:S analysis}
        \vspace{-0.5 cm}
    \end{figure}

\subsection{Impact of Multiple Paths in a Multi-User Scenario}
\label{Impact of nLOS Paths in a Multi-User Scenario}
In Section~\ref{sec: When Exploiting nLOS Paths Is Beneficial?}, we analytically investigated the benefits of exploiting \gls{nLOS} \gls{NF} links for a simple 2D scenario. Here, we present simulation results to show that there is also a benefit for a realistic 3D scenario. We begin by outlining the simulation setup, then present and analyze the corresponding results. 

\textbf{Simulation Setup:} \MD{We adopt the scenario depicted in Fig.~\ref{fig:system model} with different numbers of \glspl{MU} in different figures. For Figs.~\ref{fig:N_prop} and \ref{fig:SINR}, $K = 2$ \glspl{MU} are positioned at $[13, -13, -5]$~m and $[11, -11, -5]$~m, respectively. As the \glspl{MU} are located at approximately the same angular position from the \gls{BS}, and following the analysis in Section~\ref{sec: Channel Model for Non-ideal Surface Reflection}, we expect that exploiting the \gls{nLOS} paths enhances the \gls{SINR} for both \glspl{MU}. Fig.~\ref{fig:capacity} assumes a single \gls{MU} and Fig.~\ref{fig:SE multiuser} considers different numbers of \glspl{MU}.} One reflecting wall in the $\y-\z$ plane is considered with the following configuration: 
$\x=15$~m, $-27$~m $\leq \y\leq -17$~m, $-5$~m $\leq\z\leq5$~m. The passivity factor $\zeta$ is chosen such that the wall causes different values of losses in the specular reflection direction, i.e., $\bar{k}_1=1, 0.6, 0.2$.
The \gls{BS} is positioned at the origin $[0, 0, 0]$~m and consists of a \gls{UPA} with $N_y \times N_z = 400 \times 10$ square elements along the $\y-\z$ axes, respectively. The element spacing for the \gls{BS} is set to half the carrier wavelength. Each \gls{MU} is equipped with a single antenna, i.e., $N_\rx = 1$. The noise variance is given by $\sigma_n^2 = W N_0 N_{\rm f}$, where $N_0 = -174$~dBm/Hz is the noise power spectral density, $W = 20$~MHz is the bandwidth, and $N_{\rm f} = 6$~dB is the noise figure. 
The center frequency is $60$~GHz, with a reference path loss of $\beta = -68$~dB at $d_0 = 1$~m and path loss exponent $\eta = 2$.

\textbf{Simulation Results:} First, we verify the accuracy of the key assumption made in Section~\ref{sec: When Exploiting nLOS Paths Is Beneficial?} to simplify analytical deviation. In particular, we assumed that the contribution of the side lobes of the beam reflected from the \gls{BS} is negligible, which is valid for \glspl{BS} with large antenna arrays. In Fig.~\ref{fig:N_prop}, we plot the \gls{SMR} vs. \gls{BS} length $L_y$. As can be seen from this figure, the impact of the side lobes is negligible \gls{wrt} the main lobes for $L_y\geq 15$~cm, which is valid for the adopted system (i.e., $N_y=400$ or  $L_y=1$~m at 60~GHz).

\begin{figure}[t]
    \centering
    \includegraphics[width=0.45\textwidth]{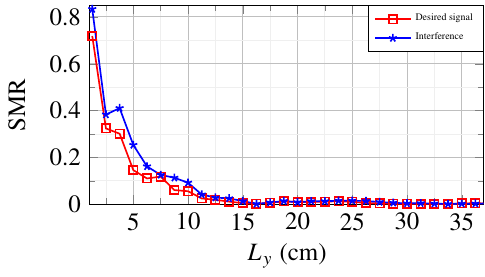}
    \caption{\gls{SMR} for the desired and interference signals vs. the \gls{BS} length ($L_y$) when $K=2$.}
    \label{fig:N_prop}
    \vspace{-0.3cm}
\end{figure}

\MD{
Fig.~\ref{fig:capacity} illustrates the ergodic capacity and an upper bound based on Jensen's inequality for a fixed covariance matrix $\bR_x=\Ex\{\bx\bx^\Herm\}$ vs. the transmit power for a single \gls{MU} \gls{MIMO} setup \cite{Goldsmith2003}. The upper bound is given by
\begin{align}
\label{eq: upper bound Jensen}
    C_\mathrm{erg}(\bR_x)&=\Ex\left\{\log_2\left(\det\left(\bI_{N_\rx}+\frac{1}{\sigma_n^2}\bH_\NF\bR_x\bH_\NF^\Herm\right)\right)\right\}\nonumber\\
    &\overset{(a)}{\leq}\log_2\left(\det\left(\bI_{N_\rx}+\frac{1}{\sigma_n^2}\Ex\left\{\bH_\NF\bR_x\bH_\NF^\Herm\right\}\right)\right),
\end{align}
where $(a)$ follows from Jensen's inequality and the fact that $\log_2\det(\cdot)$ is a concave function on the cone of positive-definite Hermitian matrices. In the considered scenario, we first assume that only the \gls{LOS} link is present. We then evaluate the setup when $S = 10$ point scatterers are additionally present, and separately when $R = 2$ extended reflectors are included alongside the \gls{LOS} path. We assume that the path coefficients $\bar{k}_r$, $\tilde{k}_r,\,\forall r\in\{1,\ldots,R\}$ and $\hat{k}_s,\,\forall s\in\{1,\ldots,S\}$ are given directly based on the physical parameters and link distances, and $N_\rx=4$ for the \gls{MU}. As observed, point scatterers only offer marginal capacity improvements in the extreme high-power regime due to severe double path-loss attenuation and their small area. In contrast, extended reflectors considerably improve capacity across the low-to-medium power regimes due to single-path-loss specular scaling. Furthermore, when both the $S = 10$ point scatterers and $R = 2$ extended reflectors are present alongside the \gls{LOS} link, the analytical Jensen upper bound remains tight and closely tracks the achievable capacity across the entire operational \gls{SNR} range. This is because the dominant channel power originates from the deterministic components, i.e., $c_0\left(\bH_\NF^\LOS+\sum_{r=1}^R\bar{k}_r\bar{\bH}_r\right)$.

\begin{figure}[htbp]
    \centering
    \includegraphics[width=0.45\textwidth]{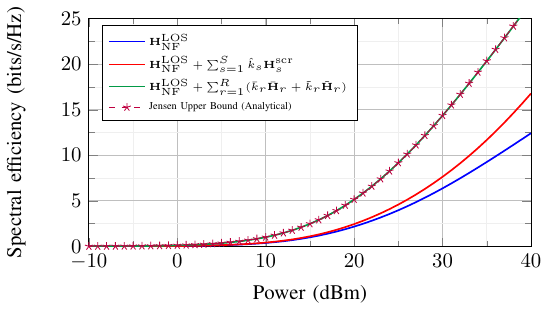}
    \caption{\MD{Ergodic channel capacity and analytical Jensen upper bound vs. the transmit power for a single \gls{MU} \gls{MIMO} setup ($K=1$), where $S=10$, $R=2$, $N_\tx=400$, $N_\rx=4$, and $\bR_x=\frac{P_t}{N_\tx}\bI_{N_\tx}$.}}
    \label{fig:capacity}
\end{figure}}

Fig.~\ref{fig:SINR} illustrates the achievable sum rate vs. the \gls{BS} transmit power for a two-\gls{MU} scenario and different wall loss values. Linear beamforming is employed at the \gls{BS}, which transmits signals along the \gls{LOS} or \gls{nLOS} paths of the \gls{BS}–\gls{MU} channel. To show the benefits of exploiting \gls{nLOS} paths for beamforming design, we compare the two scenarios considered in Section~\ref{sec: When Exploiting nLOS Paths Is Beneficial?}. From Fig.~\ref{fig:SINR}, we observe that, in the \gls{NF} regime, the large aperture of the \gls{BS} allows it to simultaneously focus energy toward both \glspl{MU}, even though they are located in nearly the same direction relative to the \gls{BS}. An additional gain is achieved when the beamforming exploits the \gls{nLOS} paths. This improvement is due to the spatial diversity provided by the multipaths. As expected from Section~\ref{sec: When Exploiting nLOS Paths Is Beneficial?}, at low power levels, exploiting only the \gls{LOS} component is preferable. On the other hand, when the transmit power is sufficiently high, exploiting the \gls{nLOS} paths enhances the \glspl{SINR} of both \glspl{MU}, based on the discussion in Section~\ref{sec: Channel Model for Non-ideal Surface Reflection}, see \eqref{eq: SINR 4} and \eqref{eq: SINR 4 2}. This leads to a higher sum-rate performance. In summary, Fig.~\ref{fig:SINR} highlights the critical role of exploiting the \gls{nLOS} components created by the \gls{NS} in enhancing performance for multi-user scenarios.

\begin{figure}[t]
    \centering
    \includegraphics[width=0.45\textwidth]{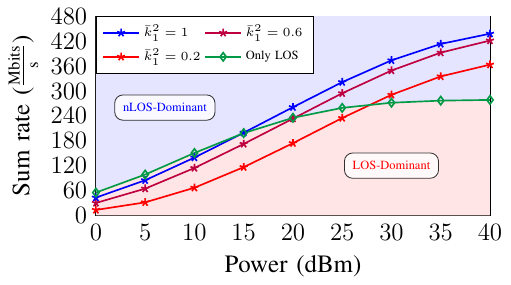}
    \caption{Sum rate vs. transmit power for $K=2$.}
    \label{fig:SINR}
    \vspace{-0.5cm}
\end{figure}

\MD{ 
In Fig.~\ref{fig:SE multiuser}, we evaluate the achievable \gls{SE} as the number of users $K$ grows, where users are uniformly distributed in the coverage region between the \gls{BS} and the \gls{NS}. As can be observed in Fig.~\ref{fig:SE multiuser}, relying exclusively on point scatterers yields negligible \gls{SE} improvement due to the severe double-path-loss attenuation and their small area. In contrast, exploiting extended reflectors (e.g., a wall or ground) substantially enhances the \gls{SE} across all user counts $K$, which confirms that the multiplexing gains observed in the two-user case scale effectively to dense multi-user scenarios.
    \begin{figure}[htbp]
        \centering
        \includegraphics[width=0.45\textwidth]{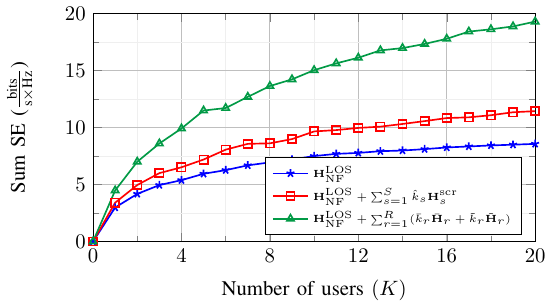}
        \caption{\MD{Achievable sum \gls{SE} vs. the number of users $K$.}}
        \label{fig:SE multiuser}
    \end{figure}
}

\section{Conclusion \MD{and Future Directions}}
\label{sec: Conclusion}
In this paper, we have introduced a novel \gls{NF} \gls{MIMO} channel model that explicitly incorporates reflections from extended \glspl{NS}, in addition to the conventional \gls{LOS} path and point scatterers. This allows a more accurate representation of realistic propagation in \gls{NF} settings. We analyzed the channel's stochastic properties and its dependence on the physical characteristics of the \glspl{NS}, such as roughness variance and surface-length correlation. The accuracy of the proposed model was rigorously validated against numerical evaluations of the \gls{HF} integral. Furthermore, we derived analytical expressions that provide insight into how \glspl{NS} contribute to the \gls{SINR} in multi-user scenarios. Based on it, our simulation results quantified the conditions under which \gls{nLOS} paths are beneficial for achieving multiplexing gain. Our analysis confirms that even though these \gls{nLOS} paths are weaker than the \gls{LOS} path, they can improve system performance with \gls{SDMA} in a multiuser scenario. The generality of our framework allows its application to other \gls{NF} \gls{MIMO} scenarios, such as those including \gls{RIS} \cite{Basar2024}, for extended target sensing in the \gls{NF} regime or for extremely large \gls{MIMO} systems in \gls{FR3} \cite{xu2025near}. For \gls{RIS} and sensing applications, the potential benefits of exploiting \gls{NF}-\gls{nLOS} paths must be carefully analyzed due to the underlying (double) path loss effect, an investigation we leave for future work. \MD{In addition, this analysis can be extended to time-varying channels \cite{matz2011fundamentals,zhang2026near} or different polarizations \cite{guo2016advances,jamali2023impact}, which is an interesting topic for future work.}

\ifarxiv

\appendices

\section{Proof of Lemma~\ref{lem: tworxtx}}
\label{app: tworxtx}
Consider (\ref{eq:correlation1}) and decouple it into the following two cases:\\
\textbf{Case 1 ($\bu'\neq\bu$):} For two different points $\bu$ and $\bu'$ on the \gls{NS},
\begin{align}
\label{eq:correlation2}
[\bR]_{(n,m,n',m')}&=\!\frac{1}{|\Uset|}\!\iint\limits_{\bu\in\Uset}\iint\limits_{\bu'\in\Uset}\!\Ex\Big\{ I(\bu,\bu_{\tx,n},\bu_{\rx,m})\Big\}\nonumber\\
&\times\Ex\Big\{ I^*\!(\bu',\bu_{\tx,n'},\bu_{\rx,m'})\Big\}\dd x'\!\dd y'\dd x\dd y,
\end{align}
resulting in $[\bR]_{(n,m,n',m')}=0$ since each expectation inside the double integrals is zero in Regime 3.\\
\textbf{Case 2 ($\bu'=\bu$):} First, let us expand $\|\bu-\bu_{\rx,m'}\|$ in terms of $\|\bu-\bu_{\rx,m}\|$. Thereby, $\|\bu-\bu_{\rx,m'}\|$ is equal to the following expression:
\begin{equation}\label{eq:29}
    \|\bu-\bu_{\rx,m}\|+\|\bu_{\rx,m}-\bu_{\rx,m'}\|\sin(\theta_\rx^l)+\bigo(\|\bu_{\rx,m}-\bu_{\rx,m'}\|^2).
\end{equation}
Note that, based on Assumption 1, we can neglect the last term in~\eqref{eq:29}. By substituting \eqref{eq:29} into \eqref{eq:correlation1}, all terms involving $\bu$ cancel out. As a result, the expression inside the expectation becomes deterministic, and the expectation operator can be omitted. A similar expansion is possible for the Tx side by substituting indices $\rx$ and $m$ with $\tx$ and $n$, respectively. Thus, the expression for $[\bR]_{(n,m,n',m')}$ simplifies as follows:
\begin{align}
 &[\bR]_{(n,m,n',m')}=\frac{1}{|\Uset|}\\&\!\!\times\!\!\!\iint\limits_{\bu\in\Uset} \e^{\jj \kappa(\|\bu_{\rx,m}-\bu_{\rx,m'}\|\sin(\theta_\rx^l)+\|\bu_{\tx,n}-\bu_{\tx,n'}\|\sin(\theta_\tx^l))}\dd x\dd y.\nonumber
\end{align}
This concludes the proof.

\section{Proof of Proposition~\ref{prop: sinc}}
\label{app: sinc}
    We begin by substituting $\bu_{\tx,n}=\bu_{\tx,n'}$ in Lemma~\ref{lem: tworxtx} and omit indices $\rx$ and $l$ for simplicity. By transforming the integral in \eqref{eq: lemma 1} from Cartesian to spherical coordinates (i.e., $\dd x\dd y=\cos(\theta)\dd\theta\dd\phi$), we obtain the following expression:
    \begin{equation}
    [\bR]_{m,m'}=\frac{1}{|\Uset|}\iint\limits_{(\phi,\theta)\in(\Phi,\Theta)}\!\!\!\! \e^{\jj \kappa\|\bu_{m}-\bu_{m'}\|\sin(\theta)}\cos(\theta)\dd\theta\dd\phi,
\end{equation}
where $(\Phi,\Theta)=\{(\phi,\theta):\phi_1\leq\phi\leq\phi_2,\theta_1\leq\theta\leq\theta_2\}$. By substituting $|\Uset|=(\phi_2-\phi_1)(\sin(\theta_2)-\sin(\theta_1))$ and continuing the integral calculation, the following result is obtained: $|[\bR]_{m,m'}|$
\begin{align} 
=&\Big|\int _{\phi_1}^{\phi_2} \int _{\theta_1}^{\theta_2} \frac{\e^{ \textsf {j}\frac {2\pi }{\lambda } \| \bu_{m'} - \bu_{m} \| \sin (\theta) }}{(\phi_2-\phi_1)(\sin(\theta_2)-\sin(\theta_1))} \cos(\theta)\dd\theta \dd\phi\Big|\nonumber \\ 
=&\Big|\e^{\jj \frac{\sin(\theta_1)+\sin(\theta_2)}{2}}\frac {\sin \left ({\frac {2\pi }{\lambda } \| \bu_{m'} - \bu_{m} \|\frac{\sin(\theta_2)-\sin(\theta_1)}{2}}\right) }{ \frac {2\pi }{\lambda } \| \bu_{m'} - \bu_{m} \|\frac{\sin(\theta_2)-\sin(\theta_1)}{2}}\Big|\nonumber\\=&\sinc\Big(\frac{2d}{\lambda}\cos\left(\frac{\theta_2+\theta_1}{2}\right)\sin\left(\frac{\theta_2-\theta_1}{2}\right)\Big).
\end{align}
This completes the proof.

\section{Proof of Lemma~\ref{lem: tworxtx correlation}}
\label{app: tworxtx correlation}
In Regime 3,  $|\bar{c}_{r}(g)|^2\to0$ in \eqref{eq:correlation1 joint}. In addition, when $C(\rho)\approx0$, e.g., for large $\rho$, \eqref{eq:correlation1 joint} is zero in Regime 3. According to Assumption 2, when the quadratic terms of 
$\rho$, i.e., $\frac{\rho^2\sin^2(\Psi_\tx)}{2\|\bu_{\tx,n'}-\bu_\rho\|}$ and $\frac{\rho^2\sin^2(\Psi_\rx)}{2\|\bu_{\rx,m'}-\bu_\rho\|}$, are non-negligible, $C(\rho)$ is approximately zero and thus \eqref{eq:correlation1 joint} is zero in Regime 3. Therefore, we can omit the quadratic terms for further calculation and have
\begin{align}
\label{eq:correlation1 joint 2}
    &\Cov\{c_r[\bH_r^\refl]_{n,m},c_r^*[\bH_r^\refl]_{n',m'}^*\}\nonumber\\
&=|c_I|^2\times\!\!\iint\limits_{\!\!\!\!\bu\in\Uset}\!\!\iint\limits_{\bu'\in\Uset}\! \e^{-\sigma_z^2\kappa_z^2(1-C(\rho))+\jj\kappa \tilde{F}}\dd x\dd y\dd x'\dd y',
\end{align}
where $\tilde{F}=\| \bu_{\tx,n}-\bu_\rho\|-\| \bu_{\tx,n'}-\bu_\rho\|-\rho\cos(\Psi_\tx)+\| \bu_{\rx,m}-\bu_\rho\|\!-\| \bu_{\rx,m'}-\bu_\rho\|\!-\rho\cos(\Psi_\rx)$. Equation~\eqref{eq:correlation1 joint 2} can be written as
\begin{align}
\label{eq:correlation1 joint 3}
    &\Cov\{c_r[\bH_r^\refl]_{n,m},c_r^*[\bH_r^\refl]_{n',m'}^*\}=|c_I|^2\nonumber\\
\times&\!\!\iint\limits_{\!\!\!\!\bu\in\Uset}\e^{\jj\kappa(\| \bu_{\tx,n}-\bu_\rho\|-\| \bu_{\tx,n'}-\bu_\rho\|+\| \bu_{\rx,m}-\bu_\rho\|\!-\| \bu_{\rx,m'}-\bu_\rho\|)}\nonumber\\
&\!\!\underbrace{\iint\limits_{\bu'\in\Uset}\! \e^{-\sigma_z^2\kappa_z^2(1-C(\rho))-\jj\kappa\rho(\cos(\Psi_\tx)+\cos(\Psi_\rx))}\dd x'\dd y'}_{\text{Not a function of $\bu_{q_1,q_2}$, $q_1=\{\tx,\rx\}$ and $q_2=\{n,m\}$}}\dd x\dd y.
\end{align}
As shown in \eqref{eq:correlation1 joint 3}, the expression consists of a product of two integrals. When computing the spatial correlation, only the normalized form is relevant, so the second integral, which is a constant, can be omitted. The first integral follows the same structure and derivation steps as for in Case~2 in Appendix~\ref{app: tworxtx}, and thus the same procedure can be applied. This concludes the proof.

\section{Proof of Theorem~\ref{theorem power regime 3}}
\label{app: power regime 3}
Substituting $\bu_{\tx,n'}=\bu_{\tx,n}$ and $\bu_{\rx,m'}=\bu_{\rx,m}$ in \eqref{eq:correlation1 joint 3} removes the first part inside the integrals. Performing a change of variables $x'-x=\rho\cos(\phi)$ and $y'-y=\rho\sin(\phi)$, where $\phi$ is the azimuth angle in the $\x-\y$ plane, $\cos(\Psi_\tx)$ and $\cos(\Psi_\rx)$ can be redefined as $\cos(\Psi_\tx)=a_\tx\cos(\phi-\phi_\tx)$ and $\cos(\Psi_\rx)=a_\rx\cos(\phi-\phi_\rx)$, where $a_\tx$ and $\phi_\tx$ ($a_\rx$ and $\phi_\rx$) are fixed and can be derived based on $\bu_{\tx,n}$ ($\bu_{\rx,m}$).  Let us first expand $a_\rx\cos(\phi-\phi_\rx)$ \gls{wrt} $\cos(\phi-\phi_\tx)$ as follows:
\begin{align}
    &a_\rx\cos(\phi-\phi_\rx)=a_\rx\cos(\phi-\phi_\tx+\phi_\tx-\phi_\rx)\nonumber\\
    =&a_\rx\cos(\phi-\phi_\tx)\cos(\phi_c)-a_\rx\sin(\phi-\phi_\tx)\sin(\phi_c),
\end{align}
where $\phi_c=\phi_\tx-\phi_\rx$ is fixed. Under this condition, the following result holds:
\begin{align}
    &\cos(\Psi_\tx)+\cos(\Psi_\rx)=a_\tx\cos(\phi-\phi_\tx)+a_\rx\cos(\phi-\phi_\rx)\nonumber\\
    =&\big(a_\tx+a_\rx\cos(\phi_c)\big)\cos(\phi-\phi_\tx)-a_\rx\sin(\phi-\phi_\tx)\sin(\phi_c)\nonumber\\
    =&A\cos(\phi-\phi_\tx-\phi_\phi),
    \label{eq: final result cosine sine}
\end{align}
where $A=\sqrt{\big(a_\tx+a_\rx\cos(\phi_c)\big)^2+a_\rx^2}$ and $\phi_\phi=\arctan(\frac{a_\rx}{a_\tx+a_\rx\cos(\phi_c)})$. Substituting \eqref{eq: final result cosine sine} into \eqref{eq:correlation1 joint 3}, $\kappa_\rho\defeq\kappa A$, and 
$\dd x'\dd y'=\rho\dd\phi\dd\rho$ yields
\begin{align}
\label{eq:correlation1 joint 4}
    &\Cov\{c_r[\bH_r^\refl]_{n,m},c_r^*[\bH_r^\refl]_{n,m}^*\}=\Ex\{|c_r|^2\}=|c_I|^2\nonumber\\
\times&\!\!\iint\limits_{\!\!\!\!\bu\in\Uset}\!\!\int\limits_{\rho}\!\int\limits_{\phi=0}^{2\pi}\! \e^{-\sigma_z^2\kappa_z^2(1-C(\rho))-\jj \kappa_\rho\rho\cos(\phi-\phi_\tx-\phi_\phi)}\rho\dd\phi\dd\rho\dd x\dd y\nonumber\\
&=2\pi|c_I|^2\!\!\iint\limits_{\!\!\!\!\bu\in\Uset}\!\!\int\limits_{\rho} \e^{-\sigma_z^2\kappa_z^2(1-C(\rho))}J_0(\kappa_\rho\rho)\rho\dd\rho\dd x\dd y,
\end{align}
where $J_0(\cdot)$ denotes the Bessel function of the first kind and zeroth order. By substituting $1-C(\rho)\approx \frac{\rho^2}{\ell^2}$ in \eqref{eq:correlation1 joint 4}, we obtain
\begin{align}
\label{eq:correlation1 joint 5}
\Ex\{|c_r|^2\}=&2\pi|c_I|^2\!\!\iint\limits_{\!\!\!\!\bu\in\Uset}\!\!\int\limits_{\rho=0}^{+\infty} \!\!\!\e^{-\sigma_z^2\kappa_z^2\frac{\rho^2}{\ell^2}}J_0(\kappa_\rho\rho)\rho\dd\rho\dd x\dd y\nonumber\\
\overset{(a)}{=}&\frac{B\ell^2}{2\sigma_z^2\kappa_z^2}\e^{-\frac{(\kappa_\rho\ell)^2}{(2\kappa_z\sigma_z)^2}}\!\!,
\end{align}
where $B$ is a constant, and $(a)$ holds when $\ell$ is not too large; otherwise, the integral diverges. As $\ell$ increases, we expect $\Ex\{|c_r|^2\}$ to increase as well. However, the expression in \eqref{eq:correlation1 joint 5} is not monotonically increasing \gls{wrt} $\ell$. Therefore, we define $\ell_{\max}=2\sigma_z \frac{\kappa_z}{\kappa_\rho}$, since the function increases up to this point. To determine $\ell_{\max}$, we computed the derivative of the function \gls{wrt} $\ell$ and set it to zero. The only positive root is given by $\ell = 2\sigma_z \frac{\kappa_z}{\kappa_\rho}$. A detailed derivation is omitted due to space constraints. This concludes the proof.
\fi

\bibliographystyle{IEEEtran}
\bibliography{References}

\end{document}